\documentclass[10pt, two column, twoside]{IEEEtran}
\usepackage{amssymb}
\usepackage{amsmath} 
\usepackage[lined,boxed,commentsnumbered, ruled]{algorithm2e}
\usepackage{mathrsfs}
\usepackage{algorithmic}
\usepackage{bm}
\usepackage{tikz}
\usetikzlibrary{arrows}
\usepackage{subfigure}
\usepackage{graphicx,booktabs,multirow}
\usepackage{epstopdf}
\usepackage{subfigure}
\usepackage{multirow}
\usepackage{tabularx} 
\usepackage{booktabs}

\definecolor{colorhkust}{RGB}{20,43,140}
\definecolor{colortsinghua}{RGB}{116,52,129}
\definecolor{color1}{RGB}{128,0,0}

\usepackage{amsthm}

\newtheorem{lemma}{Lemma}

\newtheorem{definition}{Definition}

\newcommand{\diagg}{\mathrm{diag}}

\begin{document}

        \title{Reconfigurable Intelligent Surface Assisted Massive MIMO with Antenna Selection}
\author{Jinglian He, \textit{Student Member}, \textit{IEEE}, Kaiqiang Yu, \textit{Student Member}, \textit{IEEE}, Yuanming Shi, \textit{Member}, \textit{IEEE},$\ \ \ $ Yong Zhou, \textit{Member}, \textit{IEEE},
        Wei Chen, \textit{Senior Member}, \textit{IEEE}, and Khaled B. Letaief, \textit{Fellow}, \textit{IEEE}
        \thanks{J. He is with the School of Information Science and Technology, ShanghaiTech University, Shanghai 201210, China, Shanghai Institute of Microsystem and Information Technology, Chinese Academy of Sciences, Shanghai 200050, China, and also with the University of Chinese Academy of Sciences, Beijing 100049, (e-mail: hejl1@shanghaitech.edu.cn).}
        \thanks{K. Yu, Y. Zhou and Y. Shi are with the School of Information Science and Technology, ShanghaiTech University, Shanghai 201210, China (e-mail: yukaiqiangsdu@gmail.com, zhouyong@shanghaitech.edu.cn, shiym@shanghaitech.edu.cn).}
        \thanks{W. Chen is with the Department of Electronic Engineering, Tsinghua
                University, Beijing 100084, China (e-mail: wchen@tsinghua.edu.cn).}
        \thanks{K. B. Letaief is with the Department of Electronic and Computer Engineering, Hong Kong University of Science and Technology, Hong Kong (e-mail: eekhaled@ust.hk). He is also with Peng Cheng Laboratory, Shenzhen, China.}
        \thanks{This paper has been presented in part at the IEEE International Conference on Communication (ICC), Shanghai, China, May. 2019 \cite{2019YuStochastic}.}
}
        
        \maketitle
\begin{abstract}
Antenna selection is capable of reducing the hardware complexity of massive multiple-input multiple-output (MIMO) networks at the cost of certain performance degradation. 
Reconfigurable intelligent surface (RIS) has emerged as a cost-effective technique that can enhance the spectrum-efficiency of wireless networks by reconfiguring the propagation environment. 
By employing RIS to compensate the performance loss due to antenna selection, in this paper we propose a new network architecture, i.e., RIS-assisted massive MIMO system with antenna selection, to enhance the system performance while enjoying a low hardware cost. This is achieved by maximizing the channel capacity via joint antenna selection and passive beamforming while taking into account the cardinality constraint of active antennas and the unit-modulus constraints of all RIS elements. However, the formulated problem turns out to be highly intractable due to the non-convex constraints and coupled optimization variables, for which an alternating optimization framework is provided, yielding antenna selection and passive beamforming subproblems. The computationally efficient submodular optimization algorithms are developed to solve the antenna selection subproblem under different channel state information assumptions. The iterative algorithms based on block coordinate descent are further proposed for the passive beamforming design by exploiting the unique problem structures. Experimental results will demonstrate the algorithmic advantages and desirable performance of the proposed algorithms for RIS-assisted massive MIMO systems with antenna selection.
\end{abstract}

\begin{IEEEkeywords}
Reconfigurable intelligent surface, massive MIMO, antenna selection, passive beamforming, stochastic submodular maximization.
\end{IEEEkeywords}
\section{Introduction}
\label{sec:intro}
To meet the rapidly growing traffic demand for integrated intelligent services \cite{letaief2019roadmap}, massive multiple-input multiple-output (MIMO) is recognized as a key enabling technology for future wireless communication systems \cite{larsson2014massive}. 
Equipped with a very large number of antennas at a base station (BS), massive MIMO holds the potential for dramatically increasing the spatial degrees of freedom, thereby significantly enhancing the spectral-efficiency and energy-efficiency \cite{rusek2012scaling}, as well as supporting massive connectivity \cite{liu2018massive}. 
However, each antenna needs to be supported by a dedicated radio frequency (RF) chain, which results in the high hardware cost and energy consumption. This becomes one of the key limitations for the practical implementation of massive MIMO systems and can be alleviated by a promising approach known as antenna selection \cite{mendoncca2019antenna, 2018aSimple, chen2019intelligent}. 
Specifically, a subset of antennas are selected to be connected to a small number of RF chains via an RF switching network, thereby reducing the cost and power consumption of RF chains. 

To achieve a favorable balance between system performance and hardware complexity, the authors in \cite{mendoncca2019antenna} proposed a greedy algorithm based on the matching pursuit technique to perform antenna selection, with an objective to minimize the mean square error of signal reception, while reducing the transmit power. A simple greedy algorithm based on the
submodularity and monotonicity was proposed in \cite{2018aSimple} to maximize the downlink sum-rate capacity under antenna selection constraints. The signal-to-noise
ratio (SNR) and energy efficiency maximization algorithms  were developed in \cite{gkizeli2014maximum} and \cite{li2014energy}, respectively, under the antenna selection framework. 
Although the best antennas for enhancing the spectral or energy efficiency can be found, antenna selection inevitably introduces performance loss as only a subset of antennas are active \cite{gao2017massive}. It is thus desirable to design a new network architecture that alleviates the performance loss while enjoying low hardware cost and power consumption with BS antenna selection. 

To achieve this goal, we propose to adopt the recently proposed reconfigurable intelligent surface (RIS), which provides a cost-effective way to improve the system performance by dynamically programming the wireless propagation environment \cite{wu2019towards, 2020WangZ}. Specifically, RIS is a planar meta surface consisting of many low-cost passive reflecting elements (e.g., phase shifter or printed dipoles) connected to a smart software controller \cite{Intelligent2019reflecting}. Due to the thin films form, RIS can be easily deployed onto the walls of high-rise buildings with a low deployment cost \cite{yu2020robust}. By leveraging the recent advancement of  meta materials \cite{cui2014coding}, each passive reflecting element of RIS is able to  independently
adjust its reflection coefficient for the incident signals via adjusting its reflection coefficient (i.e., passive beamforming coefficient). This  can be exploited to enhance the signal power and mitigate the performance loss due to antenna selection via effective passive beamforming design \cite{Intelligent2019reflecting}.

RIS-assisted massive MIMO systems with antenna selection can thus provide a principled way to improve system performance while reducing the hardware cost (i.e., antenna selection and low-cost RIS). This is achieved by the joint design of passive beamforming at the RIS and antenna selection at the BS, which yields the following unique challenges. Specifically, we consider the
channel capacity maximization problem, while taking into account the cardinality
constraint of the total number of active antennas and the unit-modular constraints
of all reflecting elements at the RIS. This yields a mixed combinatorial optimization problem with coupled optimization variables. 
In addition, it is generally difficult to  obtain the instantaneous and perfect channel state information (CSI) in RIS-assisted massive MIMO systems, because of the cascaded propagation channel and a large number of BS antennas \cite{yuan2020reconfigurable}. 
We thus also study the ergodic capacity maximization problem for RIS-assisted massive MIMO systems, where only historically collected channel samples are available.     

\subsection{Contributions} 
In this paper, we propose a novel network architecture, i.e., RIS-assisted massive MIMO systems with antenna selection, to enhance the spectrum efficiency while reducing the hardware cost. 
In particular, we jointly optimize the antenna selection at the BS and the phase shifts at the RIS to maximize the instantaneous/ergodic sum capacity under different CSI assumptions. 
The main contributions of this paper are summarized as follows: 

\begin{itemize}
        \item 
        We formulate a channel capacity maximization problem for RIS-assisted massive MIMO systems via joint antenna selection and passive beamforming. We propose an alternating optimization framework to decouple the optimization problem into two subproblems, i.e., the subproblem of antenna selection at BS  and the subproblem of passive beamforming at RIS. 
        \item 
        With perfect instantaneous CSI, we develop a greedy algorithm with $(1-1/e)$ approximation ratio for antenna selection by leveraging the submodular optimization technique. An iterative low-complexity algorithm with an \textit{optimal} solution in the \textit{closed-form} for passive beamforming is also provided by exploiting its unique  problem structures.
        \item 
        The ergodic sum capacity maximization problem is considered without a prior knowledge of the underlying channel distribution. We propose to solve the problem based only on  the historically collected channel samples, supported by the alternating optimization procedure, yielding the stochastic antenna selection and passive beamforming subproblems.  
        
        \item We rewrite the stochastic antenna selection subproblem as a stochastic submodular maximization problem via exploiting the submodularity and monotonicity of the objective function, followed by developing an effective stochastic gradient method with a fast gradient estimate algorithm.
        A novel iterative algorithm based on a block coordinate descent is further developed to solve the   nonconvex stochastic passive
        beamforming subproblem.    
\end{itemize}

Extensive simulation results are provided to demonstrate the excellent performance of our proposed advanced alternating optimization algorithms compared with the system without RIS and the system with only antenna selection or passive beamforming. Under the consideration of both perfect CSI and historical channel realizations, the algorithmic advantages and the desirable performance of channel capacity maximization in RIS-assisted massive MIMO systems are presented.

\subsection{Organization}
The remainder of this paper is organized as follows. We present the system architecture and problem formulation in Section \ref{sys}. The system designs with perfect CSI and channel  realizations are considered in Section \ref{section3} and Section \ref{section4}, respectively. Simulation results are illustrated in Section \ref{simulation}. Finally, we conclude this paper in Section \ref{conclusion}.

\emph{Notations}: We use boldface lowercase (e.g., $\bm{h}$) and uppercase letters (e.g., $\bm{H}$) to represent vectors and matrices, respectively. $|\cdot|$, $(\cdot)^\natural$, $\arg\{\cdot\}$ denote the absolute value, conjugate, and angle of a complex number, respectively. For a set $\mathcal{S}$, the symbols $|\mathcal{S}|$ denotes the basis of set $\mathcal{S}$. And the symbols $(\cdot)^{\sf H}$ denotes the conjugate transpose. $\mathbb{C}^{x\times y}$ denotes the space of $x\times y$ complex-value matrices, while $\bm{I}_{N\times N}$ denotes the $N\times N$ identity matrix. $\mathbb{E}[\cdot]$ denotes the statistical expectation. The diagonal matrix is denoted as $\text{diag}(\cdot)$. For a differentiable function $F(\cdot)$, we use $\nabla F$ to denote its gradient. We use $e$ to denote Euler's number. We summarize the main notations in this paper as shown in Table \ref{tabel1}.

\begin{table*}[t]
        \centering
        \caption{SYMBOL NOTATIONS}\label{tabel1} 
        \vspace{0mm}
        \resizebox{\textwidth}{!}{
                \linespread{1.4}\selectfont
                \begin{tabular}{|c|l|c|l|}
                        \hline  
                        \textbf{Symbol} & \textbf{Description} & \textbf{Symbol} & \textbf{Description} \\
                        \hline
                        $L$ &  Total number of BS antennas & $\hat{\bm{H}}(\mathcal{S})$ & Direct channel matrix from selected BS antennas to users \\
                        \hline  
                        $\mathcal{L}$ & Index set of BS antennas & $\bm{T}(\mathcal{S})$ & Channel matrix from selected BS antennas to RIS \\ \hline  
                        $N_S$ & Number of selected active BS antennas & $\bm{R}$ & Channel matrix from RIS to users\\  
                        \hline  
                        $\mathcal{S}$ & Index set of selected active antennas & $\bm{\Theta}$ & RIS reflection matrix\\  
                        \hline 
                        $N$ & Number of RIS reflecting elements & $C_{\bm{H}}(\mathcal{S},\bm{\Theta})$ & Downlink sum capacity\\  
                        \hline
        \end{tabular}}
        \vspace{-5mm}
\end{table*}

\section{System Model and Problem Formulation}\label{sys}
\subsection{System Model}
\begin{figure}[t]
        \centering
        \includegraphics[scale = 0.26]{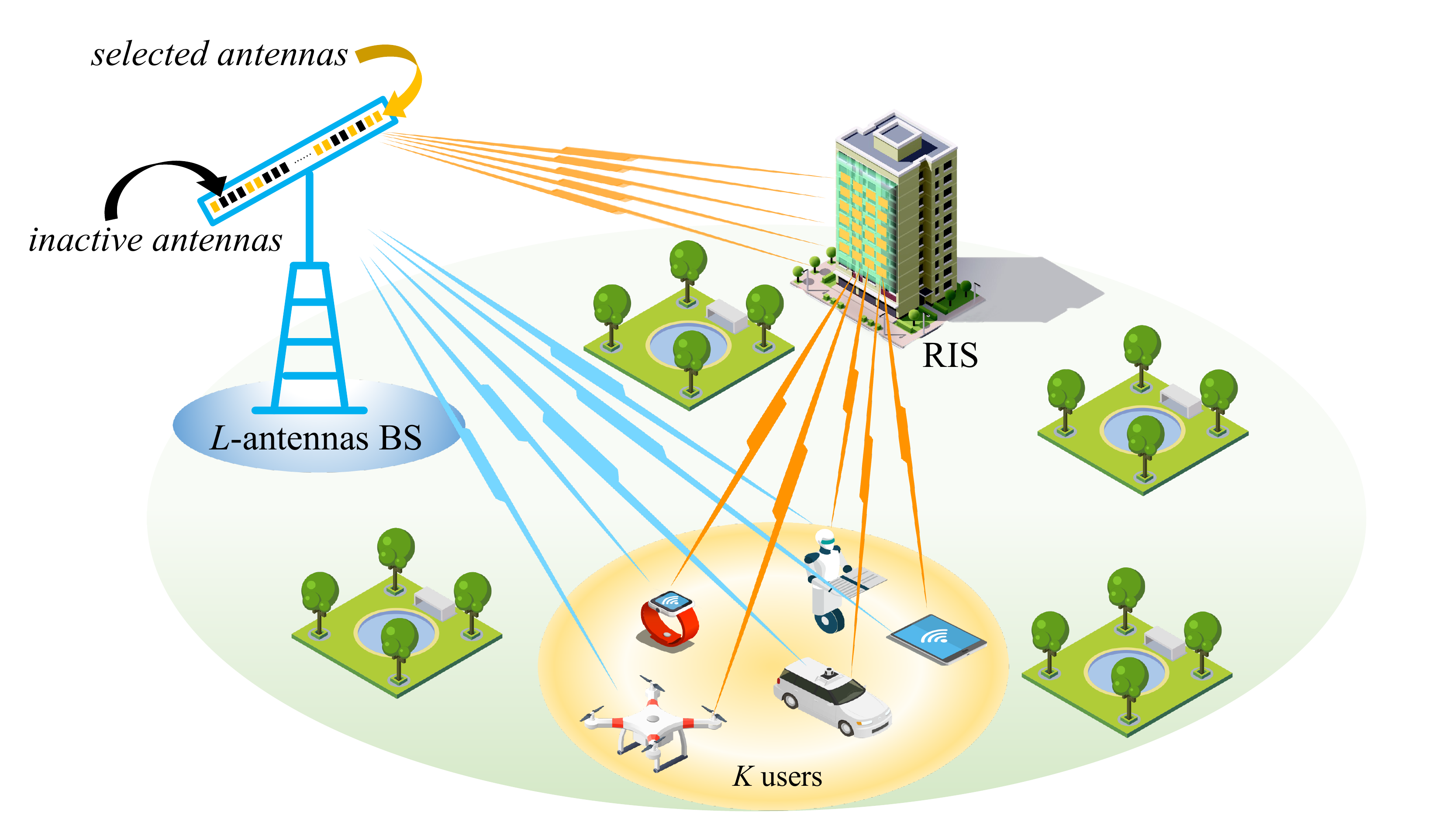}
        \caption{An RIS-assisted massive MIMO system with antenna selection.}
        \label{model}
\end{figure}
We consider a downlink massive MIMO communication system consisting of an $L$-antenna BS and $K$ single-antenna mobile users, as shown in Fig. 1, where an RIS equipped with $N$ passive reflecting elements is deployed to enhance the communication performance. We denote $\mathcal{L}=\{1,\ldots,L\}$ as the index set of antennas at the BS and $\mathcal{K}=\{1,\ldots,K\}$ as the index set of mobile users. Each reflecting element of the RIS can dynamically adjust the phase shift according to the CSI. Although impressive improvements in capacity are achieved in massive MIMO communication systems, the cost and hardware complexity scale with the number of antennas. To alleviate these drawbacks, dynamically selecting antennas becomes critical for achieving a favorable balance between performance and hardware complexity \cite{2018aSimple}. Hence, we denote $\mathcal{S}\subseteq \mathcal{L}$ as the index set of selected active antennas, where $N_S=|\mathcal{S}|$ denotes the number of active antennas.

Let $\bm{\hat{H}}(\mathcal{S})=[\bm{\hat{h}}_{1}(\mathcal{S}),\ldots,\bm{\hat{h}}_{K}(\mathcal{S})]^{\sf H}\in \mathbb{C}^{K\times N_S}$, $\bm{T}(\mathcal{S})=[\bm{t}_1(\mathcal{S}),\ldots,\bm{t}_N(\mathcal{S})]^{\sf H}\in \mathbb{C}^{N\times N_S}$, and $\bm{R}=[\bm{r}_{1},\ldots,\bm{r}_{K}]^{\sf H}\in \mathbb{C}^{K\times N}$ denote the channel matrix from the selected BS antennas to the mobile users, the channel matrix from the selected BS antennas to the RIS,  and the channel matrix from the RIS to the mobile users, respectively. Due to severe path loss, we assume that the signal reflected by the RIS more than once has negligible power and thus can be ignored \cite{2019Reconfigurable}. We consider a quasi-static block-fading channel model. Thus, the effective massive MIMO channel matrix from the BS to the mobile users is given by $\bm{H}(\mathcal{S},\bm{\Theta})=[\bm{h}_1,\ldots,\bm{h}_K]^{\sf H}=\bm{\hat{H}}(\mathcal{S})+\bm{R\Theta }\bm{T}(\mathcal{S})\in \mathbb{C}^{K\times N_{\mathcal{S}}}$, where $\bm{\Theta}=\diagg(\beta_1,\ldots,\beta_N)\in \mathbb{C}^{N\times N}$ is the diagonal reflection matrix of the RIS \cite{2019zhangCapacity} and $\beta_n\in \mathbb{C}$ is the reflection coefficient of the $n$-th RIS element. We assume that $|\beta_n|=1$ and the phase of $\beta_n$ can be flexibly adjusted in $[0,2\pi)$ \cite{cui2014coding}. 

Let $\tilde{\bm x}\in \mathbb{C}^{N_S\times 1}$ denote the transmitted signal vector across the $N_{S}$ selected antennas at the BS, and satisfies $\mathbb{E}[\|\tilde{\bm x}\|_2^2]=1$. For the sake of simplicity, we assume that the transmit power per user is fixed. The signal $\bm y(\mathcal{S},\bm{\Theta}) \in \mathbb{C}^{K\times 1}$ received at mobile users is given by 
\begin{eqnarray}
\bm y(\mathcal{S},\bm{\Theta})=\sqrt{P}\bm{H}(\mathcal{S},\bm{\Theta})\tilde{\bm{x}}+\bm z,
\end{eqnarray}
where $P$ denotes the transmit power at the BS and $\bm z\in \mathbb{C}^{K\times 1}$ denotes the additive white Gaussian noise (AWGN) vector with $z_k\sim\mathcal{CN}(0,\sigma^2), \ k\in\mathcal{K}$. The downlink sum capacity is given by \cite{Multi-cell}
\begin{eqnarray}\label{capacity}
C_{\bm{H}}(\mathcal{S},\bm{\Theta})= \text{log}_2\text{det}(\bm{I}+\text{snr}\bm{H}(\mathcal{S},\bm{\Theta})^{\sf H}\bm{H}(\mathcal{S},\bm{\Theta})),
\end{eqnarray}
where ${\sf{snr}}=P/\sigma^2$ is the signal-to-noise ratio under equal {power} allocation. Note that the channel under consideration is different from the conventional massive MIMO channel without RIS. 
The capacity of conventional massive MIMO without RIS only depends on the channel matrix $\bm{\hat{H}}(\mathcal{S})$. As the RIS-assisted massive MIMO channel matrix $\bm{H}(\mathcal{S},\bm{\Theta})$ includes the RIS reflection matrix $\bm{\Theta}$ and the selected antennas set $\mathcal{S}$, the capacity given in (\ref{capacity}) depends on both $\bm{\Theta}$ and $\mathcal{S}$.

\subsection{Problem Formulation}

In this paper, we propose to enable RIS-assisted massive MIMO capacity maximization via joint antenna selection at BS and passive beamforming at RIS, while considering the cardinality constraint of the total number of active antennas and the unit-modular constraints of all RIS elements. For ease of exposition, we first focus on an ideal scenario which assumes that perfect instantaneous CSI is available at the BS. The capacity maximization problem with perfect CSI can be formulated as
\setlength\arraycolsep{2pt}
\begin{eqnarray}\label{p1}
\mathscr{P}1:\ \mathop{\text{maximize}}_{\mathcal{S}\subseteq \mathcal{L}, \bm{\Theta}} && C_{\bm{H}}(\mathcal{S},\bm{\Theta})\\
\textrm{subject to} && |\mathcal{S}|=N_S,\\
\label{p1-4}&&|\beta_n|=1, n=1,\ldots,N.
\end{eqnarray}
However, perfect instantaneous CSI is not always possible to be obtained in practice \cite{shi2015robust,shi2014optimal,love2008overview,jindal2010unified,maddah2012completely}. To address this issue, we further formulate the following ergodic sum capacity maximization problem without any prior knowledge of the underlying channel distribution
\begin{eqnarray}\label{p2}
\mathscr{L}1:\ \mathop{\text{maximize}}_{\mathcal{S}\subseteq \mathcal{L}, \bm{\Theta}} && \mathbb{E}_{\bm{H}\sim \mathcal{D}}[C_{\bm{H}}(\mathcal{S},\bm{\Theta})]\\
\textrm{subject to} && |\mathcal{S}|=N_S,\\
\label{p2-4}&&|\beta_n|=1, n=1,\ldots,N,
\end{eqnarray}
where $\mathcal{D}$ is the underlying channel distribution.
Although problem $\mathscr{P}1$ is easier to be solved than problem $\mathscr{L}1$, both of them turn out to be highly intractable non-convex optimization problems due to the joint optimization of $\mathcal{S}$ and $\bm{\Theta}$ over the non-convex uni-modular constraint and cardinality constraint. To address this challenge, we propose to optimize $\mathcal{S}$ and $\bm{\Theta}$ alternately, resulting in two subproblems including antenna selection and passive beamforming. However, both subproblems are still non-convex due to their non-convex constraints. Hence, we propose to employ submodular optimization techniques for antenna selection, and exploit the unique structures of the objective function for passive beamforming. We elaborate the motivations and challenges at the beginning of Section \ref{section3} and Section \ref{section4}, respectively.

\section{Capacity Maximization with Perfect CSI}
\label{section3}
In this section, we first propose an alternating optimization framework to divide problem $\mathscr{P}1$ into two subproblems, and then solve the resulting antenna selection and passive beamforming subproblems by exploiting their unique structures.

\subsection{Alternating Optimization Framework}

For a fixed phase-shift matrix $\bm{\Theta}$, we write problem $\mathscr{P}1$ as the following antenna selection problem
\begin{eqnarray}\label{p3}
\mathscr{P}2:\ \mathop{\text{maximize}}_{\mathcal{S}\subseteq \mathcal{L}} && C_{\bm{H}}(\mathcal{S},\bm{\Theta})\\
\textrm{subject to} && |\mathcal{S}|=N_S.
\end{eqnarray}
Note that subproblem $\mathscr{P}$2 is a combinatorial optimization problem, for which the exhaustive search method is one of the simplest approaches to find the optimal selected antenna set. However, the large search space $\mathcal{O}(L^{N_S})$ limits its practicability and scalability, especially for massive MIMO with a large number of possible antenna selection \cite{2015Massive}. Actually, subproblem $\mathscr{P}2$ is NP-hard \cite{1995Anexact}. Therefore, we cannot derive an optimal solution in polynomial time. A large number of literatures \cite{2015Massive,2015Multi-switch,2017Reduced} tried to find a suboptimal solution in polynomial time by employing convex relaxations of the feasible selected antennas set. Since the solution of their resulting convex programming problem via convex relaxations is not guaranteed to be feasible, a post-processing fractional rounding step is further incorporated to get a suboptimal solution. However, these convex relaxation approaches are still limited by the high computational complexity $\mathcal{O}(L^{3.5})$ and non-guaranteed optimality \cite{2018aSimple}. To address these limitations, we shall develop an efficient greedy algorithm with $(1-1/e)$ approximation solution via exploiting the monotone and submodular structures of the objective function in Section \ref{section3}-B.

On the other hand, for a given selected antennas set $\mathcal{S}$, problem $\mathscr{P}1$ can be written as the following passive beamforming problem
\begin{eqnarray}\label{p4}
\mathscr{P}3:\ \mathop{\text{maximize}}_{\bm{\Theta}} && C_{\bm{H}}(\mathcal{S},\bm{\Theta})\\
\textrm{subject to} \label{p4-2}&&|\beta_n|=1, n=1,\ldots,N.
\end{eqnarray}
The above subproblem is still non-convex due to the non-concave objective function and non-convex uni-modular constraint. To tackle this challenge, we present the optimal solution in closed-form in Section \ref{section3}-C \cite{2019zhangCapacity}.
\subsection{Greedy Algorithm for Submodular Maximization}
In this subsection, we propose to solve the subproblem of antenna selection by leveraging submodularity and monotonicity of its objective function. Specifically, let $V$ be a ground set of objects $V:=\{v_1,\ldots,v_n\}$, and $2^V$ denote its power set.
\begin{definition}(Submodularity)
        \cite{2005Submodular}: A set function $g: 2^V\rightarrow \mathbb{R}$ is submodular if and only if, for any set $A,B\subseteq V$, we have
        \begin{eqnarray}
        g(A)+g(B)\geq g(A\cap B)+g(A\cup B).
        \end{eqnarray}
\end{definition}
Note that a favorable property of submodular functions is the non-increasing marginal gain. Specifically, we define the marginal gain of the object $v\in V$ as $g(A\cup\{v\})-g(A)$. The marginal gain introduced by adding $v$ to $A$ does not increase when we add $v$ to $B$ with $A\subseteq B\subseteq V\setminus v$. Inspired by this \textit{diminishing returns} property \cite{tohidi2020submodularity}, various greedy algorithms were proposed to find a theoretically guaranteed suboptimal set to maximize the submodular set-functions via iteratively picking an object with maximal marginal gain until satisfying the constraints \cite{2018aSimple}. 
\begin{definition}(Monotonicity of Set Functions)
        \cite{2018aSimple}: A set function $g$ is said to be monotone if $g(A)\leq g(B)$ for all $A\subseteq B\subseteq V$.
\end{definition}
Monotonicity is another key feature of set-functions, which plays a vital role on algorithmic techniques for getting the near-optimal solution of monotone submodular maximization problems. Intuitively, it can further improve the guaranteed approximation ratio of maximizing a set-function only with the submodular structure. We show that
the channel capacity function (\ref{p3}) has both two encouraging characteristics in the following lemma.
\begin{lemma}\label{lemma1}
        The objective set function $C_{\bm{H}}(\mathcal{S},\bm{\Theta})$ (\ref{p3}) of problem $\mathscr{P}2$ is submodular and monotone with respect to $\mathcal{S}$.
\end{lemma}
\begin{proof}
        Please refer to Appendix A.
\end{proof}
Based on Lemma \ref{lemma1}, we can reformulate problem $\mathscr{P}2$ as a submodular maximization problem under the cardinality constraint, thereby yielding a discrete greedy approach. To be specific, it starts from the empty set $\mathcal{S}_0=\varnothing$, and then incrementally adds an element $x \notin \mathcal{S}_{i-1}$ with maximal marginal gain to construct $\mathcal{S}_{i}$ at the $i$-th iteration. Mathematically, the incremental construction rule is given by \cite{2010Greedysensor}:
\begin{eqnarray}
\mathcal{S}_{i+1}\!=\!\mathcal{S}_{i}\!\cup\! \left\{\mathop{\arg\max}_{x\notin \mathcal{S}_{i}}\{C_{\bm{H}}(\mathcal{S}_{i}\!\cup\! \{x\},\bm{\Theta})
\!-\!C_{\bm{H}}(\mathcal{S}_{i},\bm{\Theta})\}\right\},\ \ \
\end{eqnarray}
where $i=0,\ldots,N_S-1$. We summarize the greedy algorithm to solve problem $\mathscr{P}$2 in Algorithm \ref{alg1}, and further employ $\mathcal{S}^{*}$ to denote the solution of the selected antennas set.
\begin{algorithm}
        \label{alg1}
        \SetKwData{Left}{left}\SetKwData{This}{this}\SetKwData{Up}{up}
        \SetKwInOut{Input}{Input}\SetKwInOut{Output}{output}
        \Input{$\bm{\hat{H}}$, $\bm{R}$, $\bm{T}$, $\bm{\Theta}$, $\mathcal{L}$, $\mathcal{S}_0=\varnothing$.}
        \For{$i =1,2,\ldots, N_S$}{
                $x\leftarrow \mathop{\arg\max}\limits_{x\notin \mathcal{S}_{i-1}}C_{\bm{H}}(\mathcal{S}_{i-1}\cup \{x\},\bm{\Theta})
                -C_{\bm{H}}(\mathcal{S}_{i-1},\bm{\Theta})$;\\
                $
                \mathcal{S}_{i}\leftarrow\mathcal{S}_{i-1}\cup \{x\}$;
        }
        \Output{$\mathcal{S}^{*}\leftarrow\mathcal{S}_{N_S}$}
        \caption{Greedy Algorithm for Problem  $\mathscr{P}$2.}
        \label{algo1}
\end{algorithm}

Clearly, the proposed greedy algorithm only needs $O(N_{S}L)$ measurements of the objective function $C_{\bm{H}}(\mathcal{S},\bm{\Theta})$, which is theoretically and practically more efficient than convex relaxation approaches with complexity $O(L^{3.5})$. Moreover, the quality of the suboptimal solution can be guaranteed according to the following lemma \cite{1978AnAnalysis}.
\begin{lemma}
        The proposed greedy algorithm as shown in Alg. \ref{alg1}, for the formulated monotone submodular maximization problem under a cardinality constraint (\ref{p3}) yields a suboptimal solution $\mathcal{S}^*$ with $(1-1/e)$ approximation ratio for any problem instance. That is,
        \begin{eqnarray}
        C_{\bm{H}}(\mathcal{S}_{opt},\bm{\Theta})\geq C_{\bm{H}}(\mathcal{S}^*,\bm{\Theta})\geq (1-1/e)C_{\bm{H}}(\mathcal{S}_{opt},\bm{\Theta}),\ \ \ 
        \end{eqnarray}
        where $\mathcal{S}_{opt}$ denotes the optimal solution of problem $\mathscr{P}2$ \cite{1978AnAnalysis}.
\end{lemma}
Note that the problem of obtaining a better worst-case approximation guarantee for problem $\mathscr{P}2$ is also NP-hard, which principally demonstrates that our proposed greedy algorithm is an optimal polynomial-time approximation algorithm \cite{1978Best}. 

\subsection{Passive Beamforming with Perfect CSI}
\label{3c}
For the given selected antennas set $\mathcal{S}^{*}$, we shall optimize the RIS phase-shift matrix by solving problem $\mathscr{P}$3. Specifically, we propose to iteratively optimize one variable (i.e., $\beta_n$) while keeping other $N-1$ variables fixed based on the principle of block coordinate descent. We formulate a non-convex problem to optimize $\beta_{n}$ with given $\mathcal{S}^*$ and $\{\beta_{j},j\neq n\}_{j=1}^N$, and then obtain an optimal solution of the resulted subproblem in closed-form via exploiting its unique structure \cite{2019zhangCapacity}. For the ease of presentation, we further revisit the following notations adopted to formulate the subproblem. Let $\bm{R}=[\bm{\rho}_1,\ldots,\bm{\rho}_N]$ and $\bm{T}(\mathcal{S}^*)=[\bm{t}_1(\mathcal{S}^*),\ldots,\bm{t}_N(\mathcal{S}^*)]^{\sf H}$, where $\bm{\rho}_n \in \mathbb{C}^{K\times 1}$ and $\bm{t}_n(\mathcal{S}^*)\in \mathbb{C}^{N_S\times 1}$. We can obtain the following subproblem with respect to $\beta_{n}$ according to \cite{2019zhangCapacity}
\begin{subequations}\label{p3n}
        \begin{align}
        \mathop{\text{maximize}}_{\beta_n}\ & \text{log}_2\text{det}\big(\bm{P}_n(\mathcal{S}^*)\!\!+\!\!\beta_n\bm{Q}_n(\mathcal{S}^*)\!\!+\!\!\beta_n^\natural\bm{Q}_n(\mathcal{S}^*)^{\sf H}\big)\\
        \textrm{subject to}\  & \label{p3n-1}|\beta_n|=1,\end{align}
\end{subequations}
where $\bm{P}_n(\mathcal{S}^*)$ and $\bm{Q}_n(\mathcal{S}^*)$ can be further expressed as
\begin{eqnarray}
\bm{P}_n(\mathcal{S}^*)=\bm{I}_{K}+{\sf{snr}}\bm{t}_n(\mathcal{S}^*)\bm{\rho}_n^{\sf H}\bm{\rho}_n{\bm{t}_n(\mathcal{S}^*)}^{\sf H}+\ \ \ \ \ \ \ \ \ \ \ \ \ \ \ \ \ \ \ \ \ \nonumber\\
\text{snr}\Bigg(\!\!\bm{\hat{H}}(\mathcal{S^*})\!+\!\!\!\!\!\!\sum\limits_{i=1,i\neq n}^{N}\!\!\!\!\beta_i\bm{\rho}_i{\bm{t}_i(\mathcal{S}^*)}^{\sf H}\!\Bigg)\!\Bigg(\!\!\bm{\hat{H}}(\mathcal{S^*}\!)\!+\!\!\!\!\!\!\sum \limits_{i=1,i\neq n}^{N}\!\!\!\!\beta_i\bm{\rho}_i{\bm{t}_i(\mathcal{S}^*)}^{\sf H}\!\Bigg)^{\sf H}\!\!,\ \ \nonumber\\
\bm{Q}_n(\mathcal{S}^*)={\sf{snr}}\bm{\rho}_n{\bm{t}_n(\mathcal{S}^*)}^{\sf H}\Bigg(\bm{H}(\mathcal{S}^*)^{\sf H}+\sum \limits_{i=1,i\neq n}^{N}\beta_i^\natural\bm{t}_i(\mathcal{S}^*)\bm{\rho}_i^{\sf H}\Bigg).\ \ \nonumber
\end{eqnarray}

Since both $\bm{P}_n(\mathcal{S}^*)$ and $\bm{Q}_n(\mathcal{S}^*)$ are independent of $\beta_n$, the objective function of subproblem (\ref{p3n}) is concave over $\beta_n$. However, it is still intractable and non-convex due to the uni-modular constraint (\ref{p3n-1}). To address this challenge, we exploit unique structures of $\bm{P}_n(\mathcal{S}^*)$ and $\bm{Q}_n(\mathcal{S}^*)$, yielding the following optimal solution in closed-form \cite{2019zhangCapacity}
\begin{eqnarray}
\label{eqn_update}
\beta_n^{\star}=
\begin{cases}
e^{-j\arg \{\lambda_n\}},& \text{if}\ \  \text{Tr}(\bm{P}_n(\mathcal{S}^*)^{-1}\bm{Q}_n(\mathcal{S}^*))\neq 0 \\
0,& \text{otherwise},
\end{cases}
\end{eqnarray}
where $\lambda_n$ represents the only non-zero eigenvalue of $\bm{P}_n(\mathcal{S}^*)^{-1}\bm{Q}_n(\mathcal{S}^*)$. Based on the above solution, we present an algorithm for solving problem $\mathscr{P}3$, which is summarized in Alg. \ref{alg2}. Specifically, it first randomly generates $\{\beta_{n}\}_{n=1}^N$ with $|\beta_{n}|=1$ and phases of each $\beta_{n}$ following the uniform distribution over $[0,2\pi)$. Then, we iteratively update each $\beta_{n}$ with others being fixed based on (\ref{eqn_update}) until convergence.
\begin{algorithm}
        \label{alg2}
        \SetKwData{Left}{left}\SetKwData{This}{this}\SetKwData{Up}{up}
        \SetKwInOut{Input}{Input}\SetKwInOut{Output}{output}
        \Input{$\bm{\hat{H}}$, $\bm{R}$, $\bm{T}$, $\mathcal{S}^{*}$.}
        Randomly generate $\{\beta_{n}\}_{n=1}^N$ with $|\beta_{n}|=1$.\\
        \For{$n =1,2,\ldots,N$}{
                \textbf{if} {$\text{Tr}(\bm{P}_n(\mathcal{S}^*)^{-1}\bm{Q}_n(\mathcal{S}^*))\!\neq\! 0$} \textbf{then} $\beta_n^{\star}\!=\!e^{-j\arg \{\lambda_n\}}$\\
                \textbf{else} $\beta_n^{\star}=0$
        }
        If not convergence, go to Step 2; otherwise, stop.\\
        \Output{ $\bm{\Theta}^{\star}=\text{diag}(\beta_1^{\star},\ldots,\beta_N^{\star})$}
        \caption{Proposed Algorithm for Problem  $\mathscr{P}$3.}
\end{algorithm}

Note that we can obtain the optimal solution of every subproblem with respect to each $\beta_{n}$, thereby yielding non-decreasing objective values of problem $\mathscr{P}3$ over iterations. Therefore, Alg. \ref{alg2} is guaranteed to be monotonic convergence.

\section{Capacity Maximization Based on Channel Realizations}
\label{section4}
As it is generally difficult to obtain perfect CSI in the RIS-assisted massive MIMO systems due to the high channel training overhead \cite{liaskos2018new,liang2019large,huang2019holographic}. We shall propose an alternating optimization framework to divide problem $\mathscr{L}$1 into two subproblems, and then solve the resulting antenna selection and passive beamforming subproblems of problem $\mathscr{L}1$ only based on the historical channel realizations without any prior knowledge of channel distribution. To be specific, we first formulate the antenna selection subproblem via the alternating optimization framework, then reformulate the antenna selection subproblem as a stochastic submodular maximization problem, for which a scalable stochastic projected gradient algorithm is developed. To reduce the complexity, we further propose a faster gradient estimating approach. 
In Section \ref{section4}-C, we propose to solve passive beamforming subproblem of problem $\mathscr{L}1$ via exploiting its unique structure only based on the historical channel realizations.
\subsection{Alternating Optimization Framework}
We first decouple the optimization variables in the objective function $\mathbb{E}_{\bm{H}\sim \mathcal{D}}[C_{\bm{H}}(\mathcal{S},\bm{\Theta})]$ of problem $\mathscr{L}$1. For a given matrix $\bm{\Theta}$,
the antenna selection subproblem is given by
\begin{eqnarray}\label{L2}
\mathscr{L}2:\ \mathop{\text{maximize}}_{\mathcal{S}\subseteq \mathcal{L}} && \mathbb{E}_{\bm{H}\sim \mathcal{D}}[C_{\bm{H}}(\mathcal{S},\bm{\Theta})]\\
\textrm{subject to} && |\mathcal{S}|=N_S,
\end{eqnarray}
where $\mathcal{D}$ is the underlying channel distribution. Subproblem $\mathscr{L}2$ turns out to be a highly intractable stochastic combinatorial optimization problem. Note that the previous literatures on antenna selection either assume the availability of the instantaneous CSI \cite{2018aSimple} or the statistical CSI \cite{2018Enhanced}. However, it is challenging to evaluate the exact and analytic expressions of ergodic sum capacity, even when the channel distribution is available. Moreover, although the high-dimensional random matrix theory can be employed to derive the deterministic approximations for the ergodic channel capacity, it is still difficult to further optimize the complicated and approximate expression \cite{2011random}. We instead propose to solve antenna selection subproblem $\mathscr{L}2$ only based on the historical channel realizations in Section \ref{section4}-B.

On the other hand, we continue to decouple the optimization variables $\mathcal{S}$ and $\bm{\Theta}$ in the objective function $\mathbb{E}_{\bm{H}\sim \mathcal{D}}[C_{\bm{H}}(\mathcal{S},\bm{\Theta})]$ of problem $\mathscr{L}1$. For the fixed active antenna set $\mathcal{S}$, the passive beamforming subproblem of $\mathscr{L}1$ can be further formulated as follows
\begin{eqnarray}
\mathscr{L}3:\ \mathop{\text{maximize}}_{\bm{\Theta}} && \mathbb{E}_{\bm{H}\sim \mathcal{D}}[C_{\bm{H}}(\mathcal{S},\bm{\Theta})]\\
\textrm{subject to} && |\beta_n|=1, n=1,\ldots,N,
\end{eqnarray}
where $\mathcal{D}$ is an unknown channel distribution. Note that the above formulation turns out to be highly intractable and non-convex due to the complicated expression of ergodic sum capacity and non-convex constraint. The solution to passive beamforming subproblem $\mathscr{L}$3 will be explored in more detail in Section \ref{section4}-C.

\subsection{Antenna Selection Subproblem}
In this paper, we propose to solve problem $\mathscr{L}2$ only based on the historical channel realizations. Due to the submodular and monotone structures of the objective function with respect to each channel realization, we can directly extend the mentioned greedy algorithm (Alg. \ref{alg1}) to a \textit{simple greedy} algorithm for this stochastic setting.  Specifically, we collect $s$ historical channel realizations from the unknown distribution, and then turn to directly optimize the following empirical objective function
\begin{eqnarray}
\mathop{\text{maximize}}_{\mathcal{S}\subseteq \mathcal{L}}\  \frac{1}{s}\sum_{i=1}^{s}C_{\bm{H}_i}(\mathcal{S},\bm{\Theta})\ \ \ 
\textrm{subject to}\ \ |\mathcal{S}|=N_S.
\end{eqnarray} 
Since the empirical objective function is also submodular and monotone, we can further employ Alg. \ref{alg1} to obtain a $(1-1/e)$ suboptimal solution \cite{1978AnAnalysis}. However, it principally needs a great quantity of samples, which restricts its scalability and practicality for large-scale RIS-assisted massive MIMO systems.

To address the scalability issue, we propose to convert problem $\mathscr{L}2$ into the continuous domain, yielding a stochastic submodular maximization problem. Then, various continuous optimization techniques can be further utilized to design scalable algorithms with theoretical guarantees.  

\subsubsection{Overview of Continuous Submodularity and Matroid}
Before lifting problem $\mathscr{L}2$ into the continuous domain, we first revisit some useful definitions. Since a large number of literatures have considered the submodular function in discrete domains \cite{2018aSimple,2010Greedysensor}, the submodular function can be naturally extended to arbitrary lattices \cite{2005Submodular}. Thus, we have the following definition. 
\begin{definition}(Smooth Submodular Function) \cite{1982Maximizing}:
        Let $\bm{D}=D_1\times D_2\times \ldots \times D_n$ denote a subset of $\mathbb{R}_{+}^n$, where $D_i$ is a compact subset of $\mathbb{R}_{+}$. A continuous function $\Psi$: $\bm{D}\rightarrow \mathbb{R}_{+}$ is smooth submodular if and only if for all $\bm{x}$, $\bm{y}\in \bm{D}$, we have 
        \begin{eqnarray}
        \Psi(\bm{x})+\Psi(\bm{y})\geq \Psi(\bm{x}\vee \bm{y})+\Psi(\bm{x}\wedge \bm{y}),
        \end{eqnarray}
        where $\bm{x}\vee \bm{y}=\max (\bm{x},\bm{y})$ and $\bm{x}\wedge \bm{y}=\min (\bm{x},\bm{y})$.
\end{definition}
Similarly, smooth submodular functions also keep the property of diminishing returns with respect to the definition of marginal gain such that $\Psi(\bm{x}+z_ie_i)\!-\!\Psi(\bm{x})$, $e_i\!\!\in\!\!\mathbb{R}^n$ and $z_i\!\in\!\mathbb{R}_+$.
\begin{definition}{(Monotonicity of Continuous Functions)} \cite{1982Maximizing}:
        A continuous function $\Psi$ is said to be monotone, if $\Psi(\bm{x})\leq \Psi(\bm{y})$ for all $\bm{x}$, $\bm{y}\in \bm{D}$ and $\bm{x}\leq \bm{y}$. Here, $\bm{x}\leq \bm{y}$ means every element of $\bm{x}$ is less than that of $\bm{y}$.
\end{definition}
Intuitively, for monotone functions, submodularity can be further restricted to be equivalent to $\Psi(\bm{x}+z_ie_i)-\Psi(\bm{x})$ keeping non-increasing for every fixed $z_i$ and $e_i$. We combine the above two properties, and formally give the following definition.  
\begin{definition}(Smooth Monotone Submodular Function) \cite{1982Maximizing}:
        A continuous function $\Psi(y)$: $[0,1]^{X}\rightarrow \mathbb{R}$ is said to be smooth monotone submodular if $\Psi$ has the following three properties:
        \begin{itemize}
                \item The function $\Psi(y)$ has second derivatives everywhere.
                \item For $\forall j\in X$, $\frac{\partial \Psi}{\partial y_j}\geq 0$ holds everywhere. (monotone)
                \item For $\forall i, j\in\!\! X$, $\frac{\partial^2\Psi}{\partial y_i \partial y_j}\!\leq\! 0$ holds everywhere. (submodular)
        \end{itemize}
\end{definition}
Based on the above definitions, we can further observe that a smooth continuous monotone submodular function \cite{1982Maximizing} is concave along any non-negative direction vector. This kind of function with diminishing returns has been well exploited as special submodular functions called DR-submodular \cite{2017GradientMethods, 2015AGeneralization}. 
Note that maximizing a monotone submodular function $g(A)$ or $\Psi(\bm{x})$ without any constraints is trivial, yielding an optimal solution such as the ground set $V$ or $\bm{D}$. To formulate practical problems, we often solve them subject to some constraints on $A$ or $\bm{x}$, which can be described by a matroid.  
\begin{definition}(Matroid)
        \cite{oxley2006matroid}: A finite matroid $\mathcal{M}=(X,\mathcal{I})$, where $X$ is a finite set named the ground set and $\mathcal{I}$ is a collected subsets of $X$ named the independent subsets of $X$ with the following properties:
        \begin{itemize}
                \item $\varnothing$ denotes the empty set, $\varnothing \in \mathcal{I}$.
                \item For all $A\subseteq B\subseteq X$, if $B\in \mathcal{I}$ then $A \in \mathcal{I}$.
                \item If any $A, B \in \mathcal{I}$ and $|A|\leq |B|$, there exists $x \in B\backslash A$, $A+x \in \mathcal{I}$.
        \end{itemize}
\end{definition}
However, the greedy algorithm only obtains a $1/2$-approximation solution for general matroids \cite{2011Maximizing}. For some special cases of matroids such as \textit{uniform matroid}, it is possible to improve the approximation factor to $(1-1/e)$. We thus present a special matroid constraint employed in our formulation.   
\begin{definition}(Matroid Polytope) \cite{2003Combinatorial}:
        With a given matroid $\mathcal{M}=(\mathcal{X},\mathcal{I})$, the matroid polytope is defined as $P(\mathcal{M})=\{\bm{x}\geq 0:\forall \mathcal{S} \subseteq \mathcal{X}; \ \sum_{l \in \mathcal{S}}^{}x_l\leq R_{\mathcal{M}}(\mathcal{S})\}$, where $R_\mathcal{M}$ denotes the rank function of a given matroid $\mathcal{M}$ that is $R_\mathcal{M}(\mathcal{S})=\text{max}\{|I|: I \subseteq \mathcal{S}$, $I \in \mathcal{I}\}$.
\end{definition}
Note that the matroid polytope is a bounded convex body \cite{1986matroids}. Moreover, another favorable property is down-monotone. To be specific, a polytope $P\in \mathbb{R}_+^X$ is said to be down-monotone if for $\forall \bm{x},\bm{y}$ such that $0\leq \bm{x}\leq\bm{y}$ and $\bm{y}\in P$, we have $\bm{x}\in P$. For problem $\max\{\Psi(\bm{x}),\bm{x}\in P(\mathcal{M})\}$, the optimum solution $\bm{x}^*$ is guaranteed to satisfy $\sum \bm{x}^*_i=R_{\mathcal{M}}(\mathcal{S})$.

\subsubsection{Stochastic Submodular Maximization}
In this subsection, we propose to lift the discrete domain problem $\mathscr{L}$2 into the continuous domain to facilitate the scalable algorithm design, yielding a continuous submodular maximization problem. To be specific, we first define the continuous function $f: \{0,1\}^{\mathcal{L}}\rightarrow \mathbb{R}$ as $f_{\bm{H}}(\bm{\hat{x}})=C_{\bm{H}}\Big(\bigcup \limits_{\hat{x}_i=1}\{i\},\bm{\Theta}\Big)$, where $\bm{\hat{x}}$ represents a random vector in $\{0,1\}$ in which each entry denotes whether the antenna is selected, i.e., $\hat{x}_i=1$ (resp. $\hat{x}_i=0$) if the $i$-th antenna is selected (resp. not selected). Thus we can make an \textit{multilinear extension} to a continuous function $F: [0,1]^{\mathcal{L}}\rightarrow \mathbb{R}$ as follows \cite{2011Maximizing}
\begin{eqnarray}
F(\bm{x})=\mathbb{E}_{\bm{H}}[f_{\bm{H}}(\bm{\hat{x}})]
=\mathbb{E}_{\bm{H}}\![\sum \limits_{\mathcal{S}\subseteq \mathcal{L}}\!C_{\bm{H}}(\mathcal{S},\bm{\Theta})\!\prod \limits_{i\in \mathcal{S}}x_i\! \prod \limits_{j\notin \mathcal{S}}(1-x_j)],\nonumber
\end{eqnarray}
where $\bm{H}$ is drawn from a distribution $\mathcal{D}$, $\bm{x} \in [0,1]^{\mathcal{L}}$ and $\bm{\hat{x}} \in \{0,1\}^{\mathcal{L}}$. Furthermore, in vector $\bm{\hat{x}}$ each coordinate $\hat{x}_i$ is independently rounded to 1 with probability $x_i$ and 0, otherwise.

Note that $F(\bm{x})$ is the expectation of $C_{\bm{H}}(\mathcal{S},\bm{\Theta})$, where the selected antennas set $\mathcal{S}$ is determined by the probability vector $\bm{x}$. Therefore, problem $\mathscr{L}$2 can be equivalently solved by maximizing the continuous function $F(\bm{x})$. To efficiently solve problem $\mathscr{L}$2, we shall further exploit the unique properties of submodularity and monotonicity in continuous domain. Based on the definitions, we obtain the following lemma.

\begin{lemma}
        The function $F(\bm{x})$ is smooth monotone submodular \cite{1982Maximizing}.
\end{lemma}

To formulate a continuous function problem, we further need to transform the cardinality constraint on $\mathcal{S}$ to the matroid constraint on $\bm{x}$. We thus define a matroid $\mathcal{M}=(\mathcal{L}, \mathcal{S})$ with constraint $R_{\mathcal{M}}(\mathcal{S})=N_S$, where $\mathcal{L}$ is the set of all antennas and $\mathcal{S}$ denotes the selected antennas set. Therefore, problem $\mathscr{L}$2 can be equivalently rewritten as the following continuous submodular monotone maximization problem
\begin{eqnarray}\label{p6}
\mathscr{L}4:\ \mathop{\text{maximize}} && F(\bm{x})\\
\textrm{subject to} &&\label{p6-1} \bm{x}\in P(\mathcal{M}),
\end{eqnarray}
where (\ref{p6-1}) is the equivalent matroid polytope constraint. This kind of stochastic submodular maximization problem can be efficiently solved via the continuous greedy algorithm, which was first proposed in \cite{2008OptimalSTOC} for solving the submodular welfare problem and has been further discovered more recently.

To solve the above problem, We first introduce a continuous greedy algorithm to solve problem problem $\mathscr{L}4$ with a smooth monotone submodular objective function $F(\bm{x})$ and a matroid polytope constraint. The philosophy of the proposed algorithm is to iteratively move along the direction of a vector constrained by $P(\mathcal{M})$ that can maximize the local gain. Therefore, it produces a $(1-1/e)$ approximated fractional solution $\bm{x}^{*}\in P(\mathcal{M})$ to problem $\mathscr{L}4$ \cite{2008OptimalSTOC}. To be specific, it starts with the particular initial point $\bm{x}(0)=\bm{0}$ and then iteratively update vector $\bm{x}(t+\delta)$ based on the direction of the following vector $\bm{v}(\bm{x}(t))$
\begin{align}
\label{independent}
\bm{v}(\bm{x}(t))=\arg \max_{\bm{v}\in P(\mathcal{M})}(\bm{v}\cdot \nabla F(\bm{x}(t))),
\end{align}
where $\nabla F(\bm{x}(t))$ can be estimated by the random sampling method presented in Section \ref{section4}-B-3). Moreover, we further observe that problem (\ref{independent}) is a linear optimization problem over $P(\mathcal{M})$. We thus can obtain $\bm{v}(\bm{x}(t))$ via finding a maximum-weight independent set in matroid, which can be easily solved. Then, the update rule of vector $\bm{x}(t+\delta)$ is given by
\begin{align}
\bm{x}(t+\delta)=\bm{x}(t)+\delta \cdot \bm{1}_{I(t)},
\end{align}
where $t\in [0,1]$ is the finite index of iterations, $\delta$ denotes the step size and  $I(t)$ is the maximum-weight independent set with respect to problem (\ref{independent}). Intuitively, the trajectory for $\{\bm{x}(t)\}_{t=0}^1$ can be regarded as a convex linear combination of vectors $\{\bm{v}(\bm{x}(t))\}_{t=0}^1$, by which we can imply the theoretical guarantee. For the ease of presentation, we have to omit more algorithmic details and summarize the continuous greedy algorithm in Alg. ${\ref{algo3}}$. Note that the solution $\bm{x}(1)$ yielded by Alg. ${\ref{algo3}}$ is fractional, for which the \textit{pipage rounding} procedure presented in Section \ref{section4}-B-3) needs to be further adopted. 

\begin{algorithm}\label{algo3}
        \caption{Continuous Greedy Algorithm for Problem $\mathscr{L}4$.}
        \SetKwData{Left}{left}\SetKwData{This}{this}\SetKwData{Up}{up}
        \SetKwInOut{Input}{Input}\SetKwInOut{Output}{output}
        \Input{Rank $N_{S}:=R_{\mathcal{M}}(\mathcal{L})$ and step size $\delta=\frac{1}{9N_{S}^2}$.}
        {\textbf{Initialization :}} $t=0$, $\bm{x}(0)=\bm{0}$.\\
        \While{$t<1$}{
                Obtain $\bm{\phi}(t)$ via $\frac{10}{\delta^2}(1+\text{ln}(L))$ independent samples, where $\mathbb{E}[\bm{\phi}(t)|\bm{x}(t)]=\nabla F(\bm{x}(t))$;\\
                $I(t)\leftarrow\arg \max_{\bm{v}\in P(\mathcal{M})}(\bm{v}\cdot \bm{\phi}(t))$;\\
                $\bm{x}(t+\delta)\leftarrow\bm{x}(t)+\delta \cdot \bm{1}_{I(t)}$;\\
                $t\leftarrow t+\delta$;\\
                
        }
        \Output {$\bm{x}(1)$.}
\end{algorithm}

Although the continuous greedy algorithm can solve problem $\mathscr{L}4$ with a $(1-1/e)$ approximated solution, it is necessary to start with a specific initial vector $\bm{0}$ \cite{2011Maximizing}. Moreover, it needs a huge fixed batch samples to estimate gradient $\nabla F(\bm{x}(t))$, yielding the high iteration cost.
Hence, we will present a scalable stochastic projected gradient method in Section \ref{section4}-B-3).

\subsubsection{Stochastic Projected Gradient Method}

In this subsection, we shall employ the stochastic projected gradient method (SPGM) to solve the reformulated problem $\mathscr{L}4$ with the strong approximation guarantee to the global maxima \cite{2017GradientMethods}. Basically, the stochastic projected gradient method is also a greedy algorithm, which iteratively updates the decision vector $\bm{x}$ based on the estimated gradients of $F(\bm{x})$ instead of maximal marginal gains. To be specific, it starts from an arbitrary initial estimate $\bm{x}^1 \in P(\mathcal{M})$. Then, the iterative update rule is given as 
\begin{align}\label{updaterule}
\bm{x}^{t+1}=\mathcal{P}_{P(\mathcal{M})}(\bm{x}^{t}+\mu_t \nabla F(\bm{x}^t)),
\end{align}
where $\mathcal{P}_{P(\mathcal{M})}(\bm{x})$ denotes the Euclidean projection of $\bm{x}$ onto the set $P(\mathcal{M})$ and $\mu_t$ is the step size.

Unfortunately, it is difficult to evaluate $\nabla F(\bm{x})$ without any knowledge of the underlying distribution, for which we further propose to utilize the stochastic unbiased estimate $\bm{\phi}$ of the gradient according to the collected historical realizations. Thus, the iteratively update rule (\ref{updaterule}) can be rewritten as $\bm{x}^{t+1}=\mathcal{P}_{P(\mathcal{M})}(\bm{x}^{t}+\mu_t \bm{\phi}_t)$,
where $\bm{\phi}_t$ is the unbiased estimate obtained via random sampling of historical channel realizations, following the rule $\mathbb{E}[\bm{\phi}_t|\bm{x}^t]=\nabla F(\bm{x}^t)$.
We then summarize the execution of the proposed SPGM for problem $\mathscr{L}4$ in Alg. ${\ref{algo4}}$.
\begin{algorithm}\label{algo4}
        \label{SGM}
        \caption{Stochastic Projected Gradient Method for Problem $\mathscr{L}4$.}
        \SetKwData{Left}{left}\SetKwData{This}{this}\SetKwData{Up}{up}
        \SetKwInOut{Input}{Input}\SetKwInOut{Output}{output}
        \Input{Integer $T> 0$ and scalars $\mu_t> 0,\ t\in[T]$.}
        {\textbf{Initialization :}} $\bm{x}^1\in \mathcal{P}_{P(\mathcal{M})}$.\\
        \For{$t= 1,2,\ldots,T$}{
                $\bm{u}^{t+1}\leftarrow \bm{x}^t+\mu_t\bm{\phi}_t$, where $\mathbb{E}[\bm{\phi}_t|\bm{x}^t]=\nabla F(\bm{x}^t)$;\\
                $\bm{x}^{t+1}\leftarrow\arg\min_{x\in P(\mathcal{M})}\|\bm{x}-\bm{u}^{t+1}\|_2$;
        }
        Select $\tau$ uniformly at random from ${1,2,...,T}$;\\
        \Output {$\bm{x}^{\tau}$.}
\end{algorithm}
\begin{lemma}
        \label{SGMT}
        Let $\bm{\phi}_t$ denote an unbiased estimate satisfying $\mathbb{E}[\bm{\phi}_t]=\nabla F(\bm{x}^t)$ and  $\mathbb{E}[\|\bm{w}_t-\nabla F(\bm{x}^t)\|_{\ell_2}^2]\leq \delta^2$. Assume that the function $F(\bm{x})$ is $L$-smooth, which means that $\|\nabla F(\bm{x})-\nabla F(\bm{y})\|_{\ell_2}\leq L\|\bm{x}-\bm{y}\|_{\ell_2}$. By executing selecting Alg. ${\ref{algo4}}$ with step size $\mu_t=\frac{1}{L+\frac{\delta}{R}\sqrt{t}}$, and randomly select $\tau$ from $\{2,...,T-1\}$ or from $\{1, T\}$ with probability $\frac{1}{T-1}$ or $\frac{1}{2(T-1)}$ respectively, we can guarantee
        \begin{eqnarray}
        \label{sgmbound}
        \mathbb{E}[F(\bm{x}^{\tau})]\geq \frac{OPT}{2}-\epsilon,
        \end{eqnarray}
        where $T=O(\frac{R^2L}{\epsilon}+\frac{R^2\delta^2}{\epsilon^2})$ is the number of iterations, $\epsilon=\left(\frac{R^2L}{2T}+\frac{R\delta}{\sqrt{T}}\right)$, $R^2=\sup_{\bm{x},\bm{y}\in\mathcal{P}}\frac{1}{2}\|\bm{x}-\bm{y}\|_{\ell_2}^2$ is the diameter of bounded convex set $P(\mathcal{M})$ and $OPT$ is the optimal value for the formulated antenna selection subproblem $\mathscr{L}$2.
\end{lemma}
\begin{proof}
        The proof of Lemma \ref{SGMT} can be found in \cite{2017GradientMethods}.
\end{proof}
To implement Alg. \ref{algo4}, we need to consider the following issues:
\begin{itemize}
        \item \textbf{Estimating}. To obtain an unbiased estimator of the gradient $\nabla F(\bm{x})$ for the proposed stochastic projected gradient method, we can sample an antenna set $\mathcal{S}$ by selecting each antenna with probability $x_i$. Then, we can get an estimate of the $i$-th partial derivate $\phi_i$ via $C_{\bm{H}}(\mathcal{S}\cup \{i\})-C_{\bm{H}}(\mathcal{S}\setminus \{i\})$, where the channel matrix $\bm{H}$ is randomly sampled from a distribution $\mathcal{D}$. We repeat the above procedure $B$ times and then take the average \cite{2011Maximizing}. However, its computational cost is still huge due to the computations of the large-scale log-determinant, for which we shall propose a faster method to avoid redundant computations in Section \ref{section4}-B-4).
        
        \item \textbf{Pipage Rounding}. Since the solution yielded by SPGM is still fractional in continuous domain, we need further incorporate a pipage rounding procedure to obtain the discrete solution of selected antennas subset $\mathcal{S}$ \cite{2004Pipage}. Moreover, as the matroid constraint (\ref{p6-1}) in SPGM is uniform matroid, we can employ the randomized pipage rounding algorithm \cite{2017Stochastic}. 
\end{itemize}

\subsubsection{Proposed Speeding up Gradient Estimating}
In  this subsection, we propose a method to speed up the gradient estimation. As mentioned before, the estimation of gradient $\nabla F(\bm{x})$ is one of the most frequent operations for all gradient-based submodular maximization algorithms. Moreover, estimating gradient $\nabla F(\bm{x})$ is also a vital important procedure in the stochastic projected gradient method \cite{2019YuStochastic}. Therefore, we propose a low complexity approach for gradient estimation to reduce the computational costs dramatically.

We first analyze the time complexity for the estimation of gradient $\nabla F(\bm{x})$. Note that for each $\bm{x}^t$, we adopt $C_{\bm{H}}(\mathcal{S}\cup \{i\},\bm{\Theta})-C_{\bm{H}}(\mathcal{S}\setminus \{i\},\bm{\Theta})$ as the unbiased estimator of the $i$-th partial derivative $\phi_i$, where the antennas subset $\mathcal{S}$ is sampled from $\mathcal{L}$ based on $\bm{x}^t$. Moreover, as it needs {$\mathcal{O}(N_{S}^3)$} to evaluate $\text{det}(\bm{I}+\text{snr}\bm{H}(\mathcal{S},\bm{\Theta})^{\sf H}\bm{H}(\mathcal{S},\bm{\Theta}))$, the time complexity of estimating $\nabla F(\bm{x})$ is {$\mathcal{O}(LN_{S}^3)$}. However, the time complexity of our proposed speeding up gradient estimating method is only {$\mathcal{O}(K^3+LK^2)$}, where $K$ denotes the number of single-antenna users and $L$ is the total number of antennas.

Specifically, according to Sylvester's Determinant theorem \cite{2003Combinatorial}, the downlink sum capacity can be further rewritten as
\begin{align}
C_{\bm{H}}(\mathcal{S},\bm{\Theta})=\text{log}_2\text{det}(\bm{I}_{K}+\text{snr}\bm{H}(\mathcal{S},\bm{\Theta})\bm{H}(\mathcal{S},\bm{\Theta})^{\sf H}).
\end{align}
We can directly reduce the time complexity of evaluating the determinant to $\mathcal{O}(N_{S}^3)$ based on the above equation, which has also been exploited in \cite{2018aSimple}. But we can further reduce the computational complexity by avoiding computing the determinant. To be specific, we can equivalently express $\bm{H}(\mathcal{S}\cup \{i\},\bm{\Theta})\bm{H}(\mathcal{S}\cup \{i\},\bm{\Theta})^{\sf H}$ as 
\begin{align}
\bm{H}(\mathcal{S}\setminus \{i\},\bm{\Theta})\bm{H}(\mathcal{S}\setminus \{i\},\bm{\Theta})^{\sf H}+{\bm{u}_{i}\bm{u}_{i}^{\sf H}},
\end{align}
{where $\bm{u}_{i}$ denotes the $i$-th column of the channel matrix from the base station to users}. We further define
\begin{align}
[\bm{G}(\mathcal{S})]_i=\bm{I}_{K}+\text{snr}\bm{H}(\mathcal{S}\setminus \{i\},\bm{\Theta})\bm{H}(\mathcal{S}\setminus \{i\},\bm{\Theta})^{\sf H}.
\end{align}
Based on the above definitions, we then propose an efficient approach to estimate gradient $\bm{\phi}$ via Lemma \ref{lemma:efficient}.
\begin{lemma}\label{lemma:efficient}
        Let $\phi_i$ be the $i$-th entry of the gradient $\bm{\phi}$, and $\phi_i=C_{\bm{H}}(\mathcal{S}\cup \{i\},\bm{\Theta})-C_{\bm{H}}(\mathcal{S}\setminus \{i\},\bm{\Theta})$ denote the unbiased estimator of gradient. We can obtain unbiased estimators by
        \begin{align}
        \label{abc}
        \phi_i=\text{\rm{log}}_2(1+{\sf{snr}}\bm{u}_{i}^{\sf H}[\bm{G}(\mathcal{S})]_{i-1}^{-1}\bm{u}_{i}).
        \end{align}
\end{lemma}
\begin{proof}
        We can rewrite $C_{\bm{H}}(\mathcal{S}\cup \{i\},\bm{\Theta})$ as follows
        \begin{eqnarray}
        C_{\bm{H}}(\mathcal{S}\cup \{i\},\bm{\Theta})=\text{log}_2\text{det}\big([\bm{G}(\mathcal{S})]_{i-1}+\text{snr}\bm{u}_{i}\bm{u}_{i}^{\sf H}\big)\ \ \ \ \ \ \ \  \ \ \nonumber\\
        =\text{log}_{2}\text{det}\big([\bm{G}(\mathcal{S})]_{i-1}\big)+\text{log}_2\big(1+\text{snr}\bm{u}_{i}^{\sf H}[\bm{G}(\mathcal{S})]_{i-1}^{-1}\bm{u}_{i}\big).\nonumber
        \end{eqnarray}
        \text{Since $\text{log}_2\text{det}\big([\bm{G}(\mathcal{S})]_{i-1}\big)=C_{\bm{H}}(\mathcal{S}\setminus \{i\},\bm{\Theta})$}, the equality of (\ref{abc}) holds.
\end{proof}
Even if we can avoid the computation of the matrix determinant by employing Lemma \ref{lemma:efficient}, the standard inversion of $[G(\mathcal{S})]_i$ also requires $\mathcal{O}(K^3)$ operations. To tackle this problem, we further incorporate the Sherman-Morrison Formula \cite{1990MatrixAnalysis}, by which we only need compute the inversion of $[G(\mathcal{S})]$ once and obtain $[G(\mathcal{S})]_i$ from $[G(\mathcal{S})]_{i-1}$ based on the following equation
\begin{eqnarray}
[G(\mathcal{S})]_i^{-1}=[G(\mathcal{S})]^{-1}_{i-1}-\frac{{\sf{snr}}[G(\mathcal{S})]^{-1}_{i-1}\bm{u}_i\bm{u}_i^{\sf{H}}[G(\mathcal{S})]_{i-1}^{-1}}{1+{\sf{snr}}\bm{u}_i^{\sf{H}}[G(\mathcal{S})]_{i-1}^{-1}\bm{u}_i}.\ \ \ \
\end{eqnarray}
In each estimation, we only need compute $[G(\mathcal{S})]^{-1}$ once in $\mathcal{O}(K^3)$ and iteratively obtain each $[G(\mathcal{S})]_i^{-1}$ in $\mathcal{O}(K^2)$. Therefore, the gradient $\phi$ can be obtained in $\mathcal{O}(LK^2)$ operations.

\subsection{Passive Beamforming Subproblem}
In this paper, we propose to solve problem $\mathscr{L}3$ only based on the collected channel realizations via directly optimizing the following empirical objective function
\begin{eqnarray} \label{L3E}
\mathop{\text{maximize}}_{\bm{\Theta}} && \frac{1}{s}\sum_{i=1}^sC_{\bm{H}}(\mathcal{S},\bm{\Theta})\nonumber\\
\textrm{subject to} && |\beta_n|=1, n=1,\ldots,N, 
\end{eqnarray}
where $s$ is the number of historical channel realizations. However, it is still non-convex due to the non-concave objective function over the phase-shift $\bm{\Theta}$ and its uni-modular constraints on $\beta_{n}$. Note that the proposed design in Section \ref{3c} is not applicable to solve problem (\ref{L3E}) due to the sum form of capacity expressions. We thus propose a novel iterative optimization algorithm via convex realization and the projection.

We propose to iteratively update each variable $\beta_{n}$ with other variables being fixed based on the principle of  block coordinate descent. To be specific, we first formulate a non-convex subproblem of optimizing $\beta_{n}$ with given $\mathcal{S}^*$ and $\{\beta_{j},j\neq n\}_{j=1}^N$. For the consistent of presentation, we continue to employ the same notations used in Section \ref{3c}. Then, we can obtain the following subproblem
\begin{eqnarray}\label{IRCUP}
\mathop{\text{maximize}}_{\beta_n}\!&&\! \frac{1}{s}\sum_{i=1}^s\!\text{log}_2\text{det}\big(\bm{P}_n^i(\mathcal{S}^*)\!+\!\beta_n\bm{Q}_n^i(\mathcal{S}^*)\!+\!\beta_n^\natural\bm{Q}_n^i(\mathcal{S}^*)^{\sf H}\big)\nonumber\\
\textrm{subject to}\!&&\! \label{IRCUPsubject}|\beta_n|=1, n=1,2,\ldots,N,
\end{eqnarray}
where $\bm{P}_n^i$ and $\bm{Q}_n^i$ can be regarded as the instances of $\bm{P}_n$ and $\bm{Q}_n$ with respect to the $i$-th sampled channel realization. Although the objective function is concave over $\beta_{n}$, it is still non-convex due to the uni-modular constraint. Note that the proposed solution in Section \ref{3c} cannot solve this problem due to the sum form of capacity expressions. Hence, the optimal solution is difficult to obtain. To overcome this drawback, we shall first solve a relaxed convex problem by assuming $|\beta_{n}|\leq 1$, and then projecting the solution $\beta_{n}^*$ to the feasible set. Formally, we can obtain the following relaxed convex problem
\begin{eqnarray}\label{L4En}
\mathop{\text{maximize}}_{\beta_n}\!&& \!\frac{1}{s}\sum_{i=1}^s\!\text{log}_2\text{det}\big(\bm{P}_n^i(\mathcal{S}^*)\!+\!\beta_n\bm{Q}_n^i(\mathcal{S}^*)\!+\!\beta_n^\natural\bm{Q}_n^i(\mathcal{S}^*)^{\sf H}\big)\nonumber\\
\textrm{subject to}&& \label{L4Ensubject}|\beta_n|\leq 1, n=1,2,\ldots,N.
\end{eqnarray}
Note that problem (\ref{L4En}) is convex and can be efficiently solved by CVX \cite{2014CVX}. Define $\mathcal{F}=\{\beta_{n}\big||\beta_{n}|=1\}$ as the feasible set, we then have the following feasible solution: $\beta_{n}^\circ=\text{Pj}_{\mathcal{F}}(\beta_{n}^*)$, where $\beta_{n}^*$ is the optimal solution of problem (\ref{L4En}) and $\text{Pj}_{\mathcal{F}}(\cdot)$ indicates the projection operation onto $\mathcal{F}$. We then summarize the execution of proposed algorithm in Alg. \ref{alg5}.
\begin{algorithm}
        \label{alg5}
        \SetKwData{Left}{left}\SetKwData{This}{this}\SetKwData{Up}{up}
        \SetKwInOut{Input}{Input}\SetKwInOut{Output}{output}
        \Input{$\{\bm{\hat{H}}_i,\bm{R}_i,\bm{T}_i\}_{i=1}^s$, $\mathcal{S}^{*}$.}
        Randomly generate $\{\beta_{n}\}_{n=1}^N$ with $|\beta_{n}|=1$.\\
        \For{$n =1,2,\ldots,N$}{
                Solve problem (\ref{L4En}) and obtain a solution $\beta_{n}^*$;\\
        }
        If not convergence, go to Step 2;\\
        ${\beta}_{n}^\circ=\text{Pj}_{\mathcal{F}}(\beta_{n}^*),\ n=1,2,\ldots,N$;\\
        \Output{ $\bm{\Theta}^\circ=\text{diag}(\beta_1^\circ,\ldots,\beta_N^\circ)$}
        \caption{Proposed Algorithm for Problem  $\mathscr{L}3$.}
\end{algorithm}

Since $\beta_{n}^*$ is the optimal solution for problem (\ref{L4En}),  we can find a local optimum of problem (\ref{L3E}) with a relaxed constraint $|\beta_{n}|\leq 1$. Note that $\bm{\Theta}^\circ$ obtained by projection is not a local optimal solution of problem (\ref{L3E}). However, it has been shown that the performance of the projection solution still highly depends on the solution of original problem \cite{huang2018energy}. Thus, our proposed iterative algorithm can still achieve good performances after the projection.

\section{Simulation Results}\label{simulation}
In this section, we provide the simulation results of the proposed algorithms for joint antenna selection and passive beamforming in RIS-assisted massive MIMO communication systems. 
We assume that the BS equipped with a uniform linear array of antennas with antenna separation $d_A=\lambda/2$ ($\lambda$ is the wavelength) is located at altitude $h_{\text{BS}}$ meter (m) and the RIS with a uniform planar array is located at altitude $h_{\text{RIS}}$ meter (m). Thus, the locations of the BS and the RIS are set as $(0,0,10)$ and $(50,50,15)$, respectively. Moreover, the users are randomly located in the region of $(200,\pm 50,0)\times (300,\pm 50,0)$ meters. We illustrate the locations of the BS, RIS, and users' horizontal projections in Fig. \ref{x-y-model}. 
\begin{figure}[t]
       \centering
       \includegraphics[scale = 0.25]{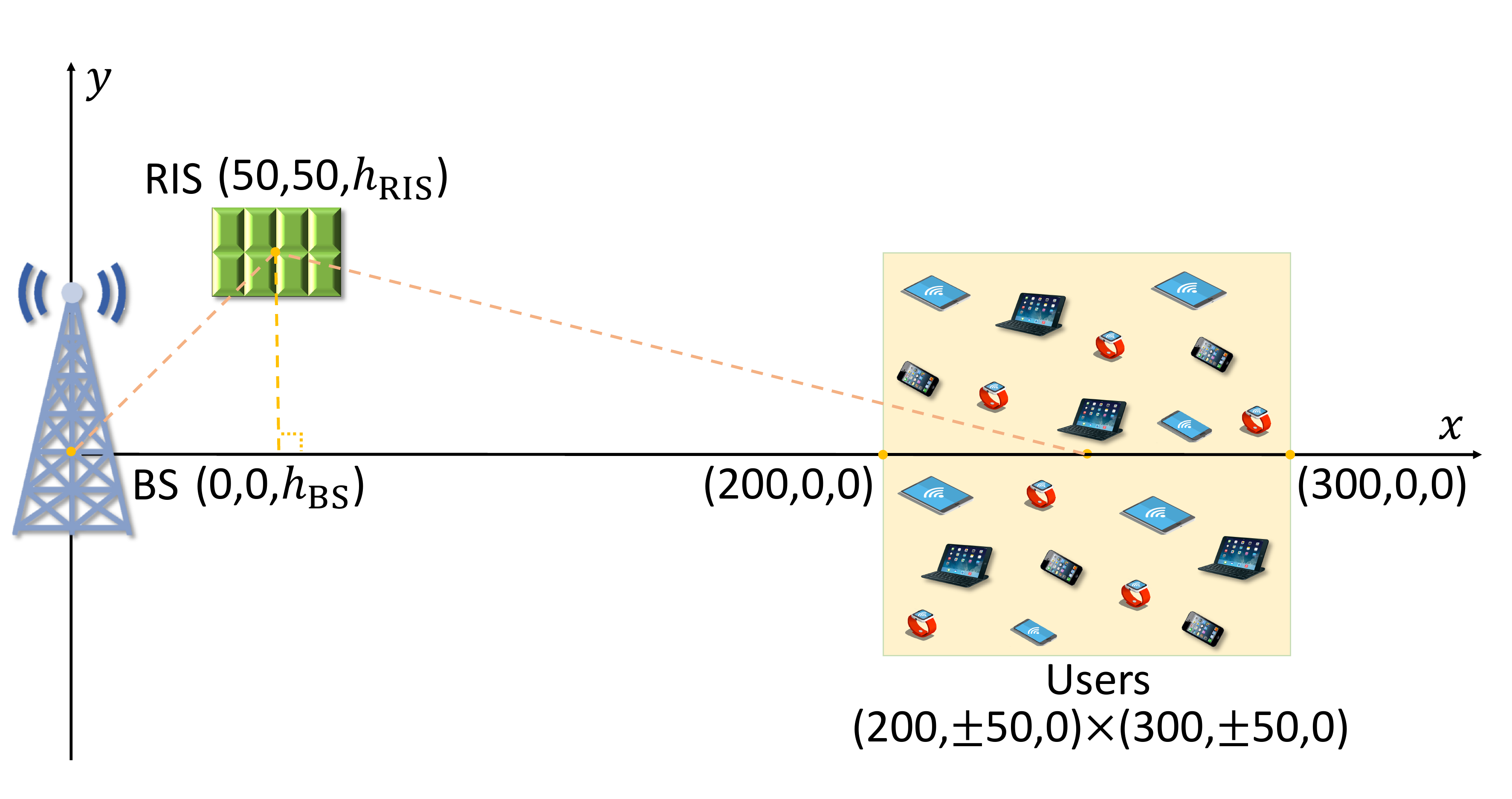}
       \caption{Horizontal locations of the BS, RIS, and users.}
       \label{x-y-model}
\end{figure}

We further consider the following path loss model
\begin{eqnarray}
L(d)=T_0(d/d_0)^{-\alpha},
\end{eqnarray}
where $T_0=-30$ dB denotes the path loss with respect to reference distance $d_0=1$ meter, $d$ represents the link distance, and $\alpha$ is the path loss exponent. The path loss exponents for the BS-user link, the BS-RIS link, and the RIS-user link are set as 3.5, 2.2, and 2.8, respectively \cite{2019zhangCapacity}. We assume that the noise power spectrum density is $-169$ dBm
/Hz with additional 9 dB noise figure, and the system bandwidth is 10MHz, yielding $\sigma^2=-90$ dBm for narrowband MIMO systems. To account for the small-scale fading, we assume that all channels suffer from Rician fading \cite{2019zhangCapacity}, i.e., $\bm{\hat{H}}, \bm{R}$, and $\bm{T}$. To be specific, the Rician fading channel can be expressed as 
\begin{eqnarray}
\bm{\tilde{H}}=\sqrt{\frac{1}{\kappa+1}}\bm{\tilde{H}}^{\text{NLoS}}+\sqrt{\frac{\kappa}{\kappa+1}}\bm{\tilde{H}}^{\text{LoS}},
\end{eqnarray}
where $\kappa \in [0,\infty)$ is the Rician factor, $\bm{\tilde{H}}^{\text{NLoS}}$ denotes the non-LoS (NLoS) component, and $\bm{\tilde{H}}^{\text{LoS}}$ denotes the deterministic line of sight (LoS) component. The LoS component $\bm{\tilde{H}}^{\text{LoS}}$ is formulated as
\begin{eqnarray}
\bm{\tilde{H}}^{\text{LoS}}=\bm{a}_{\text{D}_r}(\theta^{\text{AoA}})\bm{a}_{\text{D}_t}^{\sf H}(\theta^{\text{AoD}}),
\end{eqnarray}
where $\theta^{\text{AoA}}\in [0,2\pi)$ is the angle of arrival (AoA) and $\theta^{\text{AoD}}\in [0,2\pi)$ is the angle of departure (AoD), and
\begin{eqnarray}\label{a_r}
\bm{a}_{\text{D}_r}(\theta^{\text{AoA}})=[1,e^{j\frac{2\pi d_{\text{A}}}{\lambda}\text{sin}\theta^{\text{AoA}}},\ldots,e^{j\frac{2\pi d_\text{A}}{\lambda}(\text{D}_r-1)\text{sin}\theta^{\text{AoA}}}],
\end{eqnarray}
\begin{eqnarray}\label{a_t}
\bm{a}_{\text{D}_t}(\theta^{\text{AoD}})=[1,e^{j\frac{2\pi d_\text{A}}{\lambda}\text{sin}\theta^{\text{AoD}}},\ldots,e^{j\frac{2\pi d_\text{A}}{\lambda}(\text{D}_t-1)\text{sin}\theta^{\text{AoD}}}].
\end{eqnarray}
In (\ref{a_r}) and (\ref{a_t}), $\text{D}_r$ and $\text{D}_t$ denote the number of antennas or elements at the receiver side and transmitter side \cite{pan2020multicell}, respectively.
\begin{table*}[t]
        \centering
        \caption{SUMMARY OF SYSTEM SETUP AND SIMULATION PARAMETERS}\label{tabel112} 
        \vspace{0mm}
        \resizebox{\textwidth}{!}{
                \linespread{1.5}\selectfont
                \begin{tabular}{|c|c|c|c|}
                        \hline  
                        & Direct link, $\hat{\bm{H}}$ & BS-RIS link, $\bm{T}$ & RIS-user link, $\bm{R}$ \\
                        \hline
                        Path loss& $L_{\text{BU}}^k(d_{\text{BU}}^k)=T_0(d_{\text{BU}}^k/d_0)^{-\alpha}$ & $L_{\text{BR}}(d_{\text{BR}})=T_0(d_{\text{BR}}/d_0)^{-\alpha}$ & $L_{\text{RU}}^k(d_{\text{RU}}^k)=T_0(d_{\text{RU}}^k/d_0)^{-\alpha}$ \\
                        \hline  
                        Rician factor & $\kappa_{\text{BU}}$ & $\kappa_{\text{BR}}$ & $\kappa_{\text{RU}}$ \\ \hline  
                        AoA & $\theta_{\text{BU}}^{\text{AoA}}=0$ & $\theta_{\text{BR}}^{\text{AoA}}=\frac{\pi}{2}$ &$\theta_{\text{RU}}^{\text{AoA}}=0$\\  
                        \hline  
                        AoD & $\theta_{\text{BU}}^{\text{AoD}}=0$ & $\theta_{\text{BR}}^{\text{AoD}}=0$ &$\theta_{\text{RU}}^{\text{AoD}}=\frac{\pi}{2}$\\  
                        \hline 
                        LoS component & $\hat{\bm{H}}^{\text{LoS}}=\bm{a}_{\text{R}}(\theta_{\text{BU}}^{\text{AoA}})\bm{a}_{\text{T}}^{\sf H}(\theta_{\text{BU}}^{\text{AoD}})$ & $\bm{T}^{\text{LoS}}=\bm{a}_{\text{R}}(\theta_{\text{BR}}^{\text{AoA}})\bm{a}_{\text{T}}^{\sf H}(\theta_{\text{BR}}^{\text{AoD}})$ &$\bm{R}^{\text{LoS}}=\bm{a}_{\text{R}}(\theta_{\text{RU}}^{\text{AoA}})\bm{a}_{\text{T}}^{\sf H}(\theta_{\text{RU}}^{\text{AoD}})$\\  
                        \hline
        \end{tabular}}
        \vspace{-3mm}
\end{table*}
Let $\kappa_{\text{BU}}=20$, $\kappa_{\text{BR}}=10$ and $\kappa_{\text{RU}}=10$ be the Rician factors of the BS-users links, BS-RIS links, and RIS-users links, respectively. We further denote $d_{\text{BU}}^{k}$, $d_{\text{BR}}$, and $d_{\text{RU}}^{k}$ as the distance between user $U_k$ and the BS, between the BS and the RIS, and between user $U_k$ and the RIS, respectively. The corresponding channel coefficients are given in Table \ref{tabel112}.
All the simulation results are averaged over 100 independent channel realizations.

\subsection{Simulation Results based on Perfect CSI}
We consider the capacity maximization problem with perfect CSI in an RIS-assisted massive MIMO system, which consists of a 128-antenna BS, $K=8$ single-antenna users and an RIS with $N=50$ passive reflecting elements.

\textbf{Effectiveness.} We first study the performance of the proposed algorithms in RIS-assisted systems by showing the average achievable rate under various number of active antennas in Fig. \ref{exp1}. The average achievable rate grows as the number of active antennas increases, which indicates that more active antennas at the BS yields a better performance. In addition, it is clear that all greedy algorithms can achieve a near-optimal solution, and significantly outperform random selection methods. Moreover, greedy selection with phase optimization via proposed alternating optimization framework outperforms the greedy selection with random phase shifts, which demonstrates the necessity of joint antenna selection and passive beamforming. Note that we obtain the average achievable rate of the massive MIMO system without RIS with greedy selection by setting $\{\beta_{n}=0\}_{n=1}^N$. We can see that the RIS-assisted massive MIMO system performs better than the traditional system without RIS by comparing greedy selection with phase optimization and greedy selection without RIS. It further indicates the effectiveness of deploying the RIS in massive MIMO systems.
\begin{figure}[t]
       \centering
      \includegraphics[scale = 0.50]{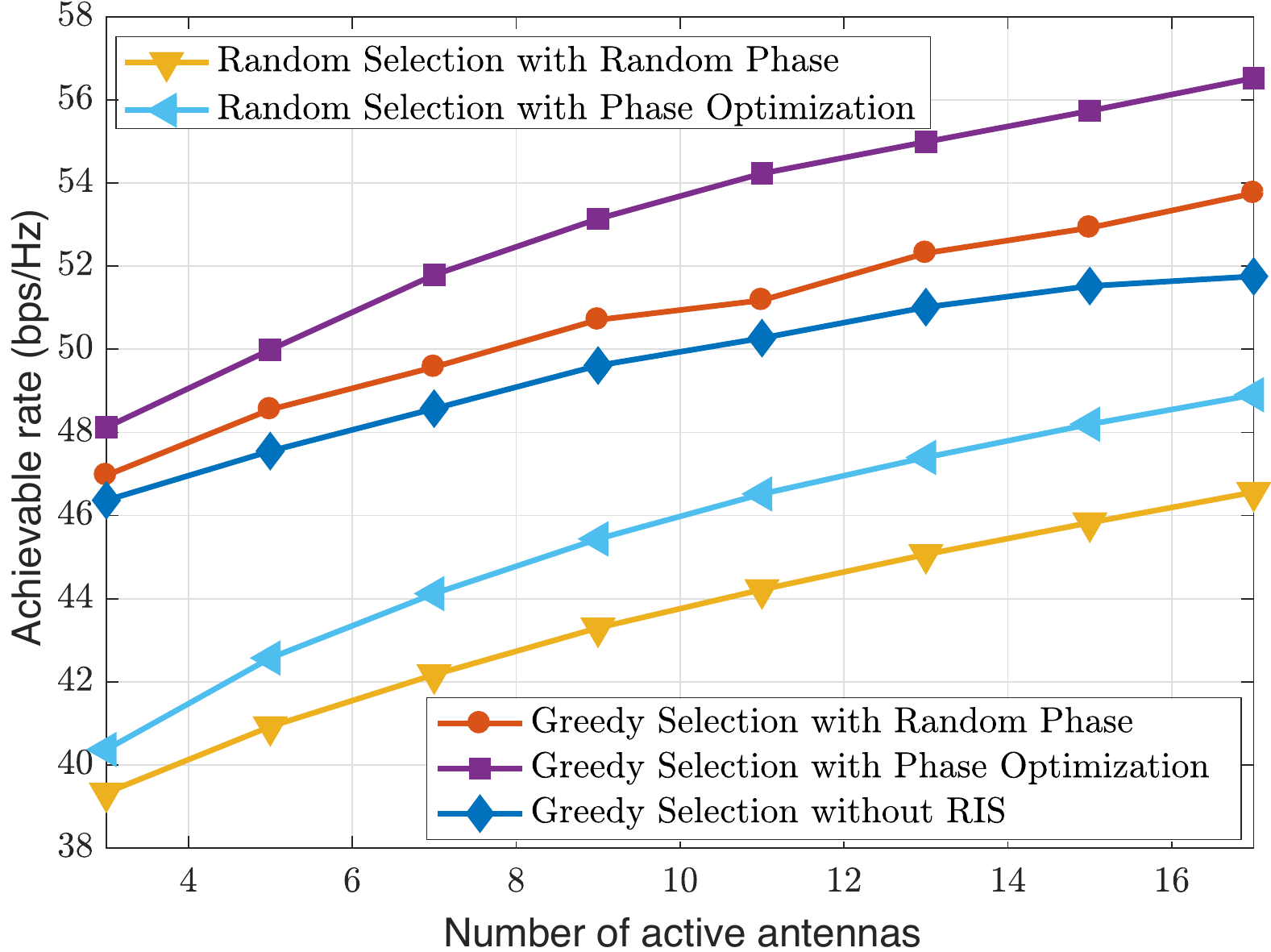}
       \caption{Achievable rate versus the number of active antennas with perfect CSI}
      \label{exp1}
\end{figure}

\textbf{Effect of varying reflecting elements.} We then investigate the impact of the number of reflecting elements on the achievable rate when $N_{S}=10$ in Fig. \ref{exp5}. 
A larger number of reflecting elements yields a higher achievable data rate. 
Moreover, the alternating framework outperforms the greedy selection with random phases under various settings. We then compare the effect of the antenna selection and the RIS deployment on the achievable rate. We further reports results of the massive MIMO system with various $N_{S}$=$10,11,12$ by setting $\{\beta_{n}=0\}_{n=1}^N$, i.e., without RIS. More active antennas are required in the massive MIMO system without RIS to achieve the same or better performance. Specifically, the traditional system with 11, 12 active antennas achieves same achievable rates with the RIS-aided system consisting of 60, 95 reflecting elements, respectively. It thus demonstrates that proposed RIS-aided system can achieve desired achievable rates by using less active antennas.

\begin{figure}[t]
       \centering
       \includegraphics[scale = 0.5]{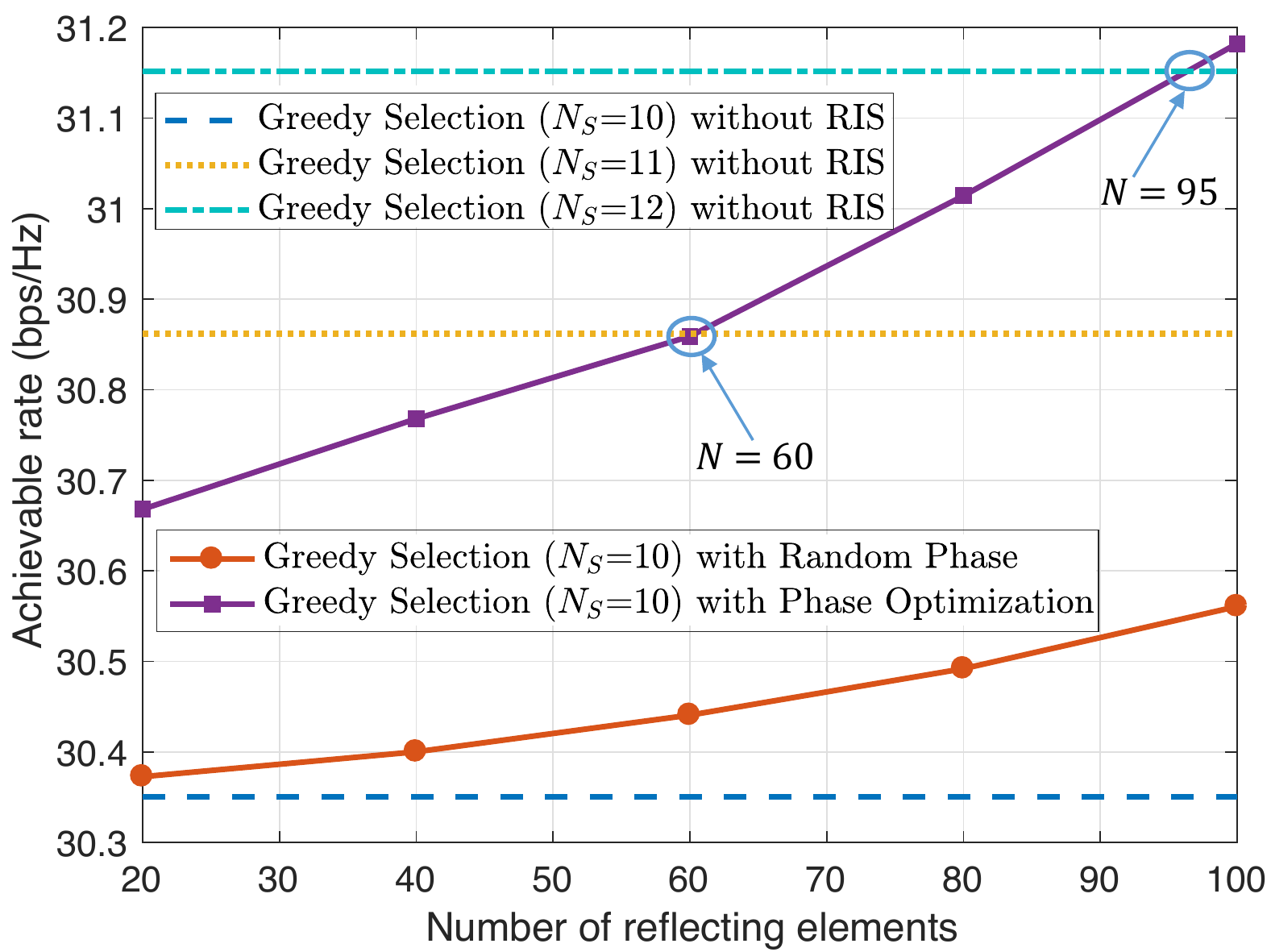}
       \caption{Achievable rate versus the number of reflecting elements with perfect CSI}
       \label{exp5}
\end{figure}

\textbf{Efficiency and scalability.} We compare the running time of greedy selection and exhaustive search for the antenna selection problem in Fig. \ref{exp6}. The exhaustive search method cannot finish in limited time (120 hours) for most settings, which limits its scalability. Moreover, the greedy selection algorithm significantly outperforms the exhaustive search method and achieves at least four orders of magnitude speedups, which demonstrates that the greedy algorithm is able to scale to large-size problems.

\begin{figure}[t]
       \centering
       \includegraphics[scale = 0.50]{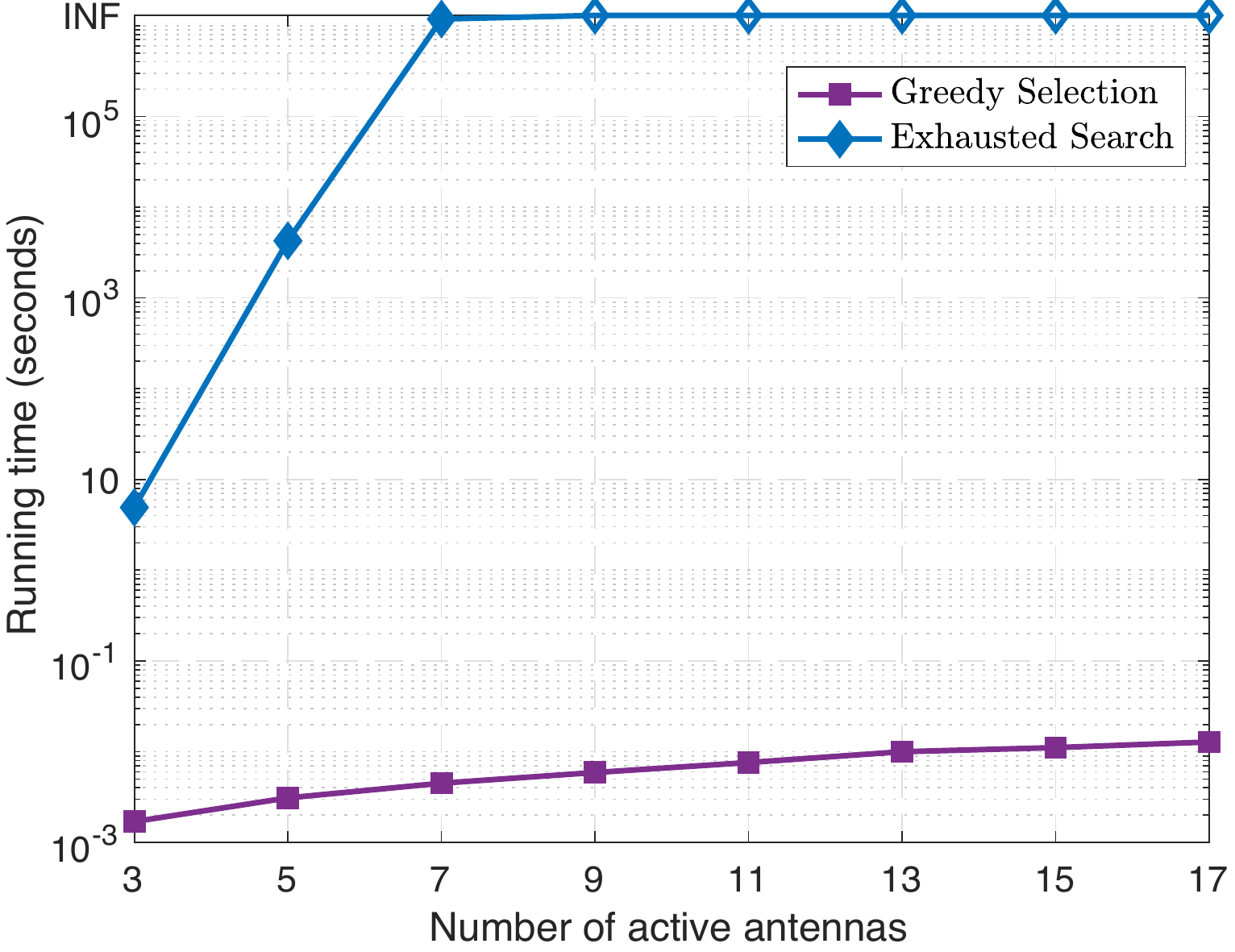}
       \caption{Running time versus the number of active antennas with perfect CSI}
       \label{exp6}
\end{figure}

\subsection{Simulation Results based on Channel Realizations}
We then consider the capacity maximization problem without any prior CSI knowledge in an RIS-assisted massive MIMO communication system with the same settings, i.e., $L=128$, $K=8$, and $N=50$. To implement the proposed algorithms, we randomly collect $s=100$ historical channel realizations by fixing all their positions. 
We plot the simulation results for the following algorithms:
\begin{itemize}
        \item[1)] \textbf{Random Selection:} randomly and independently select antennas under the matroid constraint $|\mathcal{S}|=N_{S}$.
        
        \item[2)] \textbf{Simple Greedy:} execute greedy algorithm over the empirical
        objective function with $\Lambda (N_{S} \log L+2)$ samples, where
        $\Lambda$ is an upper bound of $C_{\bm{H}}(\mathcal{S},\bm{\Theta})$ \cite{2017Stochastic}.
        
        \item[3)] \textbf{Continuous Greedy:} execute over $N_{S}^2$ iterations with
        $N_{S}^4(1 + \log L)$ samples for gradient estimate \cite{2011Maximizing}.
        
        \item[4)] \textbf{Proposed Advanced Stochastic Projected Gradient Method (Advanced SPGM):} execute with the
        step size $1/\sqrt{t}$ and $B$ samples for the speedup gradient estimate, where $t$ is the iteration index of the proposed algorithm.
        
\end{itemize}

\begin{figure}[t]
       \centering
       \includegraphics[scale = 0.50]{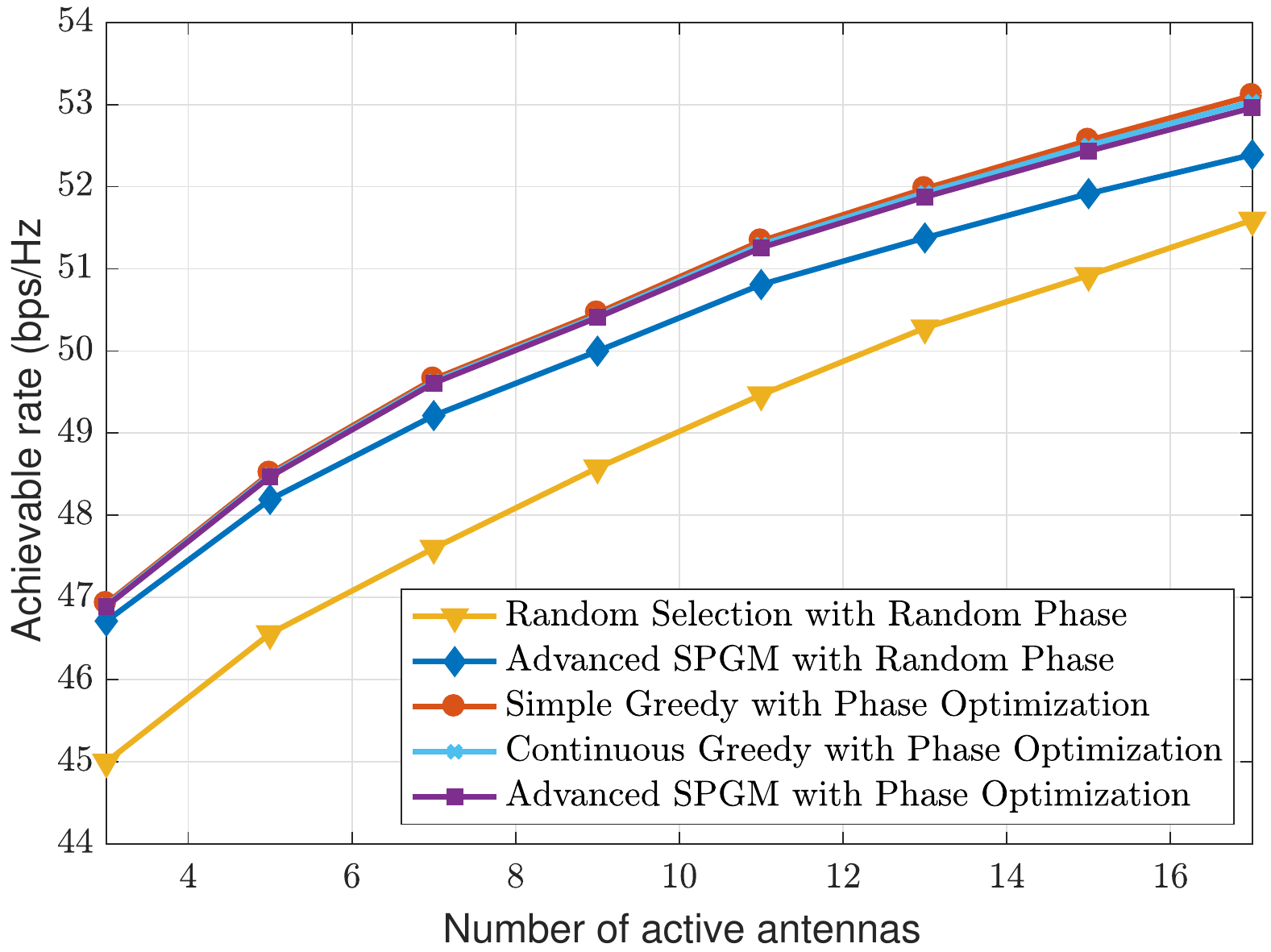}
       \caption{Achievable rate versus the number of active antennas based on channel realizations}
       \label{exp2}
\end{figure}

\begin{figure}[t]
       \centering
       \includegraphics[scale = 0.50]{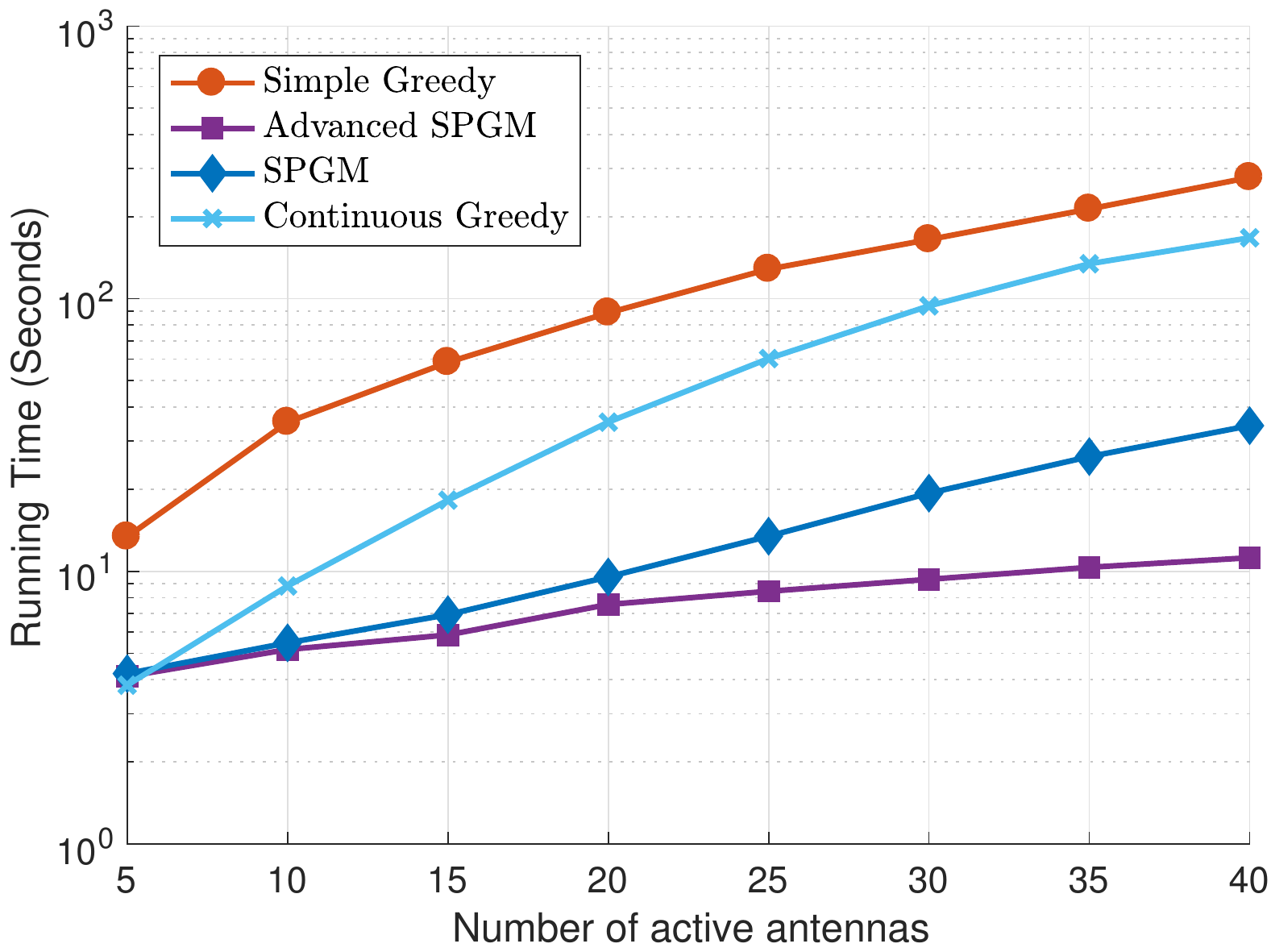}
       \caption{Running time versus the number of active antennas based on channel realizations}
       \label{exp7}
\end{figure}

\textbf{Effectiveness.}
We present the simulation results in Fig. \ref{exp2}, which clearly shows that our proposed SPGM achieves almost the same performance as the simple greedy and continuous greedy methods. Moreover, SPGM with phase optimization via proposed alternating optimization framework outperforms SPGM with random phase shifts, which shows the effectiveness of joint antenna selection and passive beamforming.

\textbf{Efficiency.} We compare different algorithms in terms of the running time in Fig. \ref{exp7}. Basically, the running time grows as the number of active antennas increases. Our proposed SPGM algorithm significantly outperforms simple greedy and continuous greedy, achieving up to 10$\times$ speedups. Moreover, the advanced SPGM algorithm with faster gradient estimating further improves the efficiency and provides up to 100$\times$ speedups.

\section{Conclusions}\label{conclusion}
In this paper, we proposed a cost-effective way and an effective algorithm framework for channel capacity maximization
via joint antenna selection and passive
beamforming in RIS-assisted massive MIMO systems. To solve the challenging system optimization problem, we proposed an alternating optimization
framework to decouple the optimization variables, resulting in antenna selection
and passive beamforming subproblems. With perfect instantaneous
CSI, we developed a greedy algorithm to solve the antenna selection
subproblem by exploiting the submodularity of its objective function. 
We also studied the scenario that any prior knowledge of channel distribution is not available, for which we proposed to solve the problem only based on historically collected
channel realizations. We leveraged the submodularity of the objective function
and reformulated the antenna selection problem as a stochastic submodular
maximization problem, followed by developing an efficient stochastic
gradient method with a faster gradient estimate. The resulting passive beamforming subproblem was solved by iteratively optimizing one variable while keeping other variables fixed, for which computationally efficient iterations were provided by exploiting the unique problem structures under different CSI assumptions. 
The experimental results showed that the proposed algorithms achieve significant performance gains and speedups.

\section{Appendix}
\subsection{Proof of Lemma 1}
Here, we will show that the objective function $C_{\bm{H}}(\mathcal{S})$ (\ref{p3}) is submodular and monotone. 

First, we show the submodularity. If all antennas of the BS are selected, it means that $N_S=L$, we can define the full channel matrix $\bm{\tilde{H}}\in \mathbb{C}^{L\times K}$, thus $\bm{H}(\mathcal{S})\in \mathbb{C}^{N_S \times K}$ is the sub-matrix of $\bm{\tilde{H}}$. Then we define a positive semi-definite matrix $\bm{Q}=\bm{\tilde{H}}^{\rm{H}}\bm{\tilde{H}}$ and let $\bm{Q}({\mathcal{S},\mathcal{S}})=\bm{H}(\mathcal{S})^{\sf H}\bm{H}(\mathcal{S}) \in \mathbb{C}^{N_S\times N_S}$ represent sub-matrix of $\bm{Q}$ with row, column indices in $\mathcal{S}\in\mathcal{L}$. Thereby,   $C_{\bm{H}}(\mathcal{S})$ can be rewritten as 
\begin{eqnarray}
\label{prove1}
C_{\bm{H}}(\mathcal{S})
=\log_2\det\left(\bm{I}_{N}(\mathcal{S},\mathcal{S})+{\sf{snr}}\bm{Q}(\mathcal{S},\mathcal{S})\right).
\end{eqnarray}
We let $\Delta:=\bm{I}_{N}+{\sf{snr}}\bm{Q}$, we can have $C_{\bm{H}}(\mathcal{S})=\log_2\det(\Delta(\mathcal{S},\mathcal{S}))$. Then, let $\bm{u}\sim \mathcal{N}(\mu,\bm{Q})$ denote multivariate Gaussian random vector in $\mathbb{C}^{L}$. We have the differential entropy of $u$ with form $f(\bm{u})=\frac{1}{2}\log\det(\Delta)$  \cite{cover2012elements}. Consider an arbitrary subset of random variables $\bm{u}_{\mathcal{S}}$ indexed by $\mathcal{S}\subseteq\mathcal{L}$, we can obtain
\begin{eqnarray}
f(\bm{u}_{\mathcal{S}})=\frac{1}{2}\log\det(\Delta(\mathcal{S},\mathcal{S}))+c|\mathcal{S}|,
\end{eqnarray}
where $c\geq 0$ is a constant independent of $\mathcal{S}$.

Specially, note that if $f(\bm{u}_{\mathcal{S}})$ is submodular, $C_{\bm{H}}(\mathcal{S})$ is also submodular. Then we will prove $f(\bm{u}_{\mathcal{S}})$ is submodular. Let us make an assumption that $\bm{u}_{\mathcal{M}}, \bm{u}_{\mathcal{N}}$ are two arbitrary subsets of random variables with $\mathcal{M},\mathcal{N}\subseteq\mathcal{S}$. Let $I(\bm{u}_{\mathcal{M}\setminus\mathcal{N}};\bm{u}_{\mathcal{N}\setminus\mathcal{M}}|\bm{u}_{\mathcal{M}\cap\mathcal{N}})$ denote the conditional mutual information, which is non-negative. It can be expressed as
 \begin{eqnarray}
I(\bm{u}_{\mathcal{M}\setminus\mathcal{N}};\bm{u}_{\mathcal{N}\setminus\mathcal{M}}|\bm{u}_{\mathcal{M}\cap\mathcal{N}})=
f(\bm{u}_{\mathcal{M}\setminus\mathcal{N}}|\bm{u}_{\mathcal{M}\cap\mathcal{N}})+\nonumber\\
f(\bm{u}_{\mathcal{N}\setminus\mathcal{M}}|\bm{u}_{\mathcal{M}\cap\mathcal{N}})-
f(\bm{u}_{\mathcal{M}\setminus\mathcal{N}},\bm{u}_{\mathcal{N}\setminus\mathcal{M}}|\bm{u}_{\mathcal{M}\cap\mathcal{N}})\nonumber\\
=f(\bm{u}_{\mathcal{M}})+f(\bm{u}_{\mathcal{N}})-f(\bm{u}_{\mathcal{M}\cup\mathcal{N}})-f(\bm{u}_{\mathcal{M}\cap\mathcal{N}})\geq 0.
\end{eqnarray}
Thus we can prove that $f(\bm{u}_{\mathcal{S}})$ is submodular according to the last inequality. Hence,  $C_{\bm{H}}(\mathcal{S})$ is also submodular. 

Second, we need show the monotonicity that $C_{\bm{H}}(\mathcal{S}\cup \{v\})\geq C_{\bm{H}}(\mathcal{S})$ for all $\mathcal{S}\subseteq\mathcal{L}$ and $v\notin \mathcal{S}$. We have
\begin{eqnarray}
C_{\bm{H}}(\mathcal{S}\cup \{v\})- C_{\bm{H}}(\mathcal{S})=\log\left(\frac{\det\left(\Delta(\mathcal{S}\cup \{v\},\mathcal{S}\cup \{v\})\right)}{\det(\Delta(\mathcal{S},\mathcal{S}))}\right).\nonumber
\end{eqnarray}
Suppose that $\Delta(\mathcal{A}\cup \{v\},\mathcal{A}\cup \{v\})\bm{x}=\bm{e}_v$, by Cramer’s rule, we can easily obtain
\begin{eqnarray}
\begin{aligned}
-\log(x_v)&=\log\left(\frac{\det\left(\Delta(\mathcal{S}\cup \{v\},\mathcal{S}\cup \{v\})\right)}{\det(\Delta(\mathcal{S},\mathcal{S}))}\right)\nonumber\\
&=-\log\left(\left(\Delta\left(\mathcal{S}\cup \{v\},\mathcal{S}\cup \{v\}\right)\right)^{-1}\bm{e}_v\right)_{v}\\
&\geq\log\lambda_{\min}(\Delta(\mathcal{S}\cup \{v\},\mathcal{S}\cup \{v\}))\\
&\geq\log\lambda_{\min}(\Delta)\\
&\geq 0.
\end{aligned}
\end{eqnarray}
Since the smallest eigenvalue of a principal submatrix is at least the smallest eigenvalue of the bigger matrix and $\Delta:=\bm{I}_{N}+{\sf{snr}}\bm{Q}\succeq \bm{I}_N$,  thus the above inequality holds. Hence, $C_{\bm{H}}(\mathcal{S})$ is monotone.

Based on the above, the objective function $C_{\bm{H}}(\mathcal{S})$ (\ref{p3}) is submodular and monotone, this proof is finished.

        \bibliographystyle{ieeetr}
        \bibliography{Reference}

\begin{thebibliography}{10}

\bibitem{2019YuStochastic}
K.~Yu, J.~He, and Y.~Shi, ``Stochastic submodular maximization for scalable
  network adaptation in dense {C}loud-{RAN},'' in {\em Proc. IEEE Int. Conf.
  Commun. (ICC)}, Shanghai, China, May. 2019.

\bibitem{letaief2019roadmap}
K.~B. Letaief, W.~Chen, Y.~Shi, J.~Zhang, and Y.-J.~A. Zhang, ``The roadmap to
  6{G}: {AI} empowered wireless networks,'' {\em IEEE Commun. Mag.}, vol.~57,
  no.~8, pp.~84--90, Aug. 2019.

\bibitem{larsson2014massive}
E.~G. Larsson, O.~Edfors, F.~, and T.~L. Marzetta, ``Massive {MIMO} for next
  generation wireless systems,'' {\em IEEE Commun. Mag.}, vol.~52, no.~2,
  pp.~186--195, Feb. 2014.

\bibitem{rusek2012scaling}
F.~Rusek, D.~Persson, B.~K. Lau, E.~G. Larsson, T.~L. Marzetta, O.~Edfors, and
  F.~Tufvesson, ``Scaling up {MIMO}: Opportunities and challenges with very
  large arrays,'' {\em IEEE Signal Process. Mag.}, vol.~30, no.~1, pp.~40--60,
  Dec. 2012.

\bibitem{liu2018massive}
L.~Liu and W.~Yu, ``Massive connectivity with massive {MIMO}—part {I}: Device
  activity detection and channel estimation,'' {\em IEEE Trans. Signal
  Process.}, vol.~66, no.~11, pp.~2933--2946, Jun. 2018.

\bibitem{mendoncca2019antenna}
M.~O. Mendon{\c{c}}a, P.~S. Diniz, T.~N. Ferreira, and L.~Lovisolo, ``Antenna
  selection in massive {MIMO} based on greedy algorithms,'' {\em IEEE Trans.
  Wireless Commun.}, vol.~19, no.~3, pp.~1868--1881, Mar. 2020.

\bibitem{2018aSimple}
A.~Konar and N.~D. Sidiropoulos, ``A simple and effective approach for transmit
  antenna selection in multiuser massive {MIMO} leveraging submodularity,''
  {\em IEEE Trans. Signal Process.}, vol.~66, no.~18, pp.~4869--4883, Aug.
  2018.

\bibitem{chen2019intelligent}
J.~Chen, S.~Chen, Y.~Qi, and S.~Fu, ``Intelligent massive {MIMO} antenna
  selection using {M}onte {C}arlo tree search,'' {\em IEEE Trans. Signal
  Process.}, vol.~67, no.~20, pp.~5380--5390, Oct. 2019.

\bibitem{gkizeli2014maximum}
M.~Gkizeli and G.~N. Karystinos, ``Maximum-{SNR} antenna selection among a
  large number of transmit antennas,'' {\em IEEE J. Sel. Topics Signal
  Process.}, vol.~8, no.~5, pp.~891--901, Oct. 2014.

\bibitem{li2014energy}
H.~Li, L.~Song, and M.~Debbah, ``Energy efficiency of large-scale multiple
  antenna systems with transmit antenna selection,'' {\em IEEE Trans. Commun.},
  vol.~62, no.~2, pp.~638--647, Jan. 2014.

\bibitem{gao2017massive}
Y.~Gao, H.~Vinck, and T.~Kaiser, ``Massive {MIMO} antenna selection: Switching
  architectures, capacity bounds, and optimal antenna selection algorithms,''
  {\em IEEE Trans. Signal Process.}, vol.~66, no.~5, pp.~1346--1360, Dec. 2017.

\bibitem{wu2019towards}
Q.~Wu and R.~Zhang, ``Towards smart and reconfigurable environment: Intelligent
  reflecting surface aided wireless network,'' {\em IEEE Commun. Mag.},
  vol.~58, no.~1, pp.~106--112, Nov. 2019.

\bibitem{2020WangZ}
Z.~Wang, Y.~Shi, Y.~Zhou, H.~Zhou, and N.~Zhang, ``Wireless-powered
  over-the-air computation in intelligent reflecting surface aided {IoT}
  networks,'' {\em IEEE Internet Things J.}, Aug. 2020.
\newblock {Early Access}.

\bibitem{Intelligent2019reflecting}
Q.~Wu and R.~Zhang, ``Intelligent reflecting surface enhanced wireless network
  via joint active and passive beamforming,'' {\em IEEE Trans. Wireless
  Commun.}, vol.~18, no.~11, pp.~5394--5409, Aug. 2019.

\bibitem{yu2020robust}
X.~Yu, D.~Xu, Y.~Sun, D.~W.~K. Ng, and R.~Schober, ``Robust and secure wireless
  communications via intelligent reflecting surfaces,'' {\em IEEE J. Sel. Areas
  Commun.}, Jul. 2020.
\newblock Early Access.

\bibitem{cui2014coding}
T.~J. Cui, M.~Q. Qi, X.~Wan, J.~Zhao, and Q.~Cheng, ``Coding metamaterials,
  digital metamaterials and programmable metamaterials,'' {\em Light: Science
  \& Applications}, vol.~3, no.~10, p.~e218, Oct. 2014.

\bibitem{yuan2020reconfigurable}
X.~Yuan, Y.-J. Zhang, Y.~Shi, W.~Yan, and H.~Liu,
  ``Reconfigurable-intelligent-surface empowered 6{G} wireless communications:
  Challenges and opportunities,'' {\em arXiv preprint arXiv:2001.00364}, 2020.

\bibitem{2019Reconfigurable}
C.~Huang, A.~Zappone, G.~C. Alexandropoulos, M.~Debbah, and C.~Yuen,
  ``Reconfigurable intelligent surfaces for energy efficiency in wireless
  communication,'' {\em IEEE Trans. Wireless Commun.}, vol.~18, no.~8,
  pp.~4157--4170, Aug. 2019.

\bibitem{2019zhangCapacity}
S.~Zhang and R.~Zhang, ``Capacity characterization for intelligent reflecting
  surface aided {MIMO} communication,'' {\em IEEE J. Sel. Areas Commun.},
  vol.~38, no.~8, pp.~1823--1838, Aug. 2020.

\bibitem{Multi-cell}
D.~Gesbert, S.~Hanly, H.~Huang, S.~S. Shitz, O.~Simeone, and W.~Yu,
  ``Multi-cell {MIMO} cooperative networks: A new look at interference,'' {\em
  IEEE J. Sel. Areas Commun.}, vol.~28, no.~9, pp.~1380--1408, Dec. 2010.

\bibitem{shi2015robust}
Y.~Shi, J.~Zhang, and K.~B. Letaief, ``Robust group sparse beamforming for
  multicast green cloud-{RAN} with imperfect {CSI},'' {\em IEEE Trans. Signal
  Process}, vol.~63, no.~17, pp.~4647--4659, Sept. 2015.

\bibitem{shi2014optimal}
Y.~Shi, J.~Zhang, and K.~B. Letaief, ``Optimal stochastic coordinated
  beamforming for wireless cooperative networks with {CSI} uncertainty,'' {\em
  IEEE Trans. Signal Process}, vol.~63, no.~4, pp.~960--973, Dec. 2014.

\bibitem{love2008overview}
D.~J. Love, R.~W. Heath, V.~K. Lau, D.~Gesbert, B.~D. Rao, and M.~Andrews, ``An
  overview of limited feedback in wireless communication systems,'' {\em IEEE
  J. Sel. Areas Commun}, vol.~26, no.~8, pp.~1341--1365, Oct. 2008.

\bibitem{jindal2010unified}
N.~Jindal and A.~Lozano, ``A unified treatment of optimum pilot overhead in
  multipath fading channels,'' {\em IEEE Trans. Commun.}, vol.~58,
  pp.~2939--2948, Oct. 2010.

\bibitem{maddah2012completely}
M.~A. Maddah-Ali and D.~Tse, ``Completely stale transmitter channel state
  information is still very useful,'' {\em IEEE Trans. Inf. Theory}, vol.~58,
  no.~7, pp.~4418--4431, Jul. 2012.

\bibitem{2015Massive}
X.~Gao, O.~Edfors, F.~Tufvesson, and E.~G. Larsson, ``Massive {MIMO} in real
  propagation environments: Do all antennas contribute equally?,'' {\em IEEE
  Trans. Commun.}, vol.~63, no.~11, pp.~3917--3928, Nov. 2015.

\bibitem{1995Anexact}
C.~W. Ko, J.~Lee, and M.~Queyranne, ``An exact algorithm for maximum entropy
  sampling,'' {\em Oper. Res.}, vol.~43, no.~4, pp.~684--691, Aug. 1995.

\bibitem{2015Multi-switch}
X.~Gao, O.~Edfors, F.~Tufvesson, and E.~G. Larsson, ``Multi-switch for antenna
  selection in massive {MIMO},'' in {\em Proc. IEEE Global Telecommun.},
  pp.~1--6, San Diego, CA, USA, Dec. 2015.

\bibitem{2017Reduced}
A.~G. Rodriguez, C.~Masouros, and P.~Rulikowski, ``Reduced switching
  connectivity for large scale antenna selection,'' {\em IEEE Trans. Commun.},
  vol.~65, no.~5, pp.~2250--2263, May. 2017.

\bibitem{2005Submodular}
S.~Fujishige, {\em Submodular functions and optimization}, vol.~58.
\newblock Amsterdam, The Netherlands: Elsevier, 2005.

\bibitem{tohidi2020submodularity}
E.~Tohidi, R.~Amiri, M.~Coutino, D.~Gesbert, G.~Leus, and A.~Karbasi,
  ``Submodularity in action: From machine learning to signal processing
  applications,'' {\em IEEE Signal Process. Mag.}, vol.~37, no.~5,
  pp.~120--133, Sept. 2020.

\bibitem{2010Greedysensor}
M.~Shamaiah, S.Banerjee, and H.Vikalo, ``sensor selection: Leveraging
  submodularity,'' in {\em Proc. IEEE Conf. Decis. Control (CDC).},
  pp.~2572--2577, Atlanta, GA, USA, Dec. 2010.

\bibitem{1978AnAnalysis}
G.~L. Nemhauser, L.~A. Wolsey, and M.~L. Fisher, ``An analysis of
  approximations for maximizing submodular set functions-{I},'' {\em Math.
  Program.}, vol.~14, no.~1, pp.~265--294, Dec. 1978.

\bibitem{1978Best}
G.~L. Nemhauser and L.~A. Wolsey, ``Best algorithms for approximating the
  maximum of a submodular set function,'' {\em Math. Oper. Res.}, vol.~3,
  no.~3, pp.~177--188, Aug. 1978.

\bibitem{liaskos2018new}
C.~Liaskos, S.~Nie, A.~Tsioliaridou, A.~Pitsillides, S.~Ioannidis, and
  I.~Akyildiz, ``A new wireless communication paradigm through
  software-controlled metasurfaces,'' {\em IEEE Commun. Mag.}, vol.~56, no.~9,
  pp.~162--169, Sept. 2018.

\bibitem{liang2019large}
Y.-C. Liang, R.~Long, Q.~Zhang, J.~Chen, H.~V. Cheng, and H.~Guo, ``Large
  intelligent surface/antennas ({LISA}): Making reflective radios smart,'' {\em
  J. Commun. Inf. Netw.}, vol.~4, no.~2, pp.~40--50, Jun. 2019.

\bibitem{huang2019holographic}
C.~Huang, S.~Hu, G.~C. Alexandropoulos, A.~Zappone, C.~Yuen, R.~Zhang,
  M.~Di~Renzo, and M.~Debbah, ``Holographic {MIMO} surfaces for 6{G} wireless
  networks: Opportunities, challenges, and trends,'' {\em IEEE Trans. Wireless
  Commun.}
\newblock Early Access.

\bibitem{2018Enhanced}
Y.~Shi, J.~Zhang, W.~Chen, and K.~B. Letaief, ``Enhanced group sparse
  beamforming for green cloud-{RAN}: A random matrix approach,'' {\em IEEE
  Trans. Wireless Commun.}, vol.~17, no.~4, pp.~2511--2524, Apr. 2018.

\bibitem{2011random}
R.~Couillet and M.~Debbah, {\em Random matrix methods for wireless
  communications}.
\newblock Cambridge Univ. Press, 2011.

\bibitem{1982Maximizing}
L.~Wolsey, ``Maximizing real-valued submodular functions: Primal and dual
  heuristics for location problems,'' {\em Math. Oper. Res.}, vol.~7,
  pp.~410--425, Aug. 1982.

\bibitem{2017GradientMethods}
S.~H. Hassani, M.~Soltanolkotabi, and A.~Karbasi, ``Gradient methods for
  submodular maximization,'' {\em Neural Inf. Process. Syst. (NIPS)},
  pp.~5843--5853, Dec. 2017.

\bibitem{2015AGeneralization}
T.~Soma and Y.~Yoshida, ``A generalization of submodular cover via the
  diminishing return property on the integer lattice,'' {\em Neural Inf.
  Process. Syst. (NIPS)}, pp.~847--855, Dec. 2015.

\bibitem{oxley2006matroid}
J.~G. Oxley, {\em Matroid theory}, vol.~3.
\newblock Oxford University Press, USA, 2006.

\bibitem{2011Maximizing}
G.~Calinescu, C.~Chekuri, M.~Pál, and J.~Vondrák, ``Maximizing a monotone
  submodular function subject to a matroid constraint,'' {\em SIAM J. Comput.},
  vol.~40, no.~6, pp.~1740--1766, Dec. 2011.

\bibitem{2003Combinatorial}
J.~Edmonds, ``Submodular functions, matroids, and certain polyhedra,'' in {\em
  Proc. Calgary Int. Conf. Combinatorial Structures and Their Applications},
  pp.~69--87, Calgary, Alta, Jun. 1969.

\bibitem{1986matroids}
N.~White, G.-C. Rota, and N.~M. White, {\em Theory of matroids}.
\newblock Cambridge Univ. Press, 1986.

\bibitem{2008OptimalSTOC}
J.~Vondrak, ``Optimal approximation for the submodular welfare problem in the
  value oracle model,'' in {\em Proc. ACM Symposium on Theory of Computing
  (STOC)}, pp.~67--74, Victoria, BC, Canada, May. 2008.

\bibitem{2004Pipage}
A.~Ageev and M.~Sviridenko, ``Pipage rounding: A new method of constructing
  algorithms with proven performance guarantee,'' {\em J. Combin. Optim.},
  vol.~8, pp.~307--328, Sept. 2004.

\bibitem{2017Stochastic}
M.~R. Karimi, M.~L., S.~H. Hassani, and A.~Krause, ``Stochastic submodular
  maximization: The case of coverage functions,'' {\em Neural Inf. Process.
  Syst. (NIPS)}, pp.~6856--6866, Dec. 2017.

\bibitem{1990MatrixAnalysis}
J.~Horn and C.~R. Johnson, {\em Matrix Analysis}.
\newblock Cambridge Univ. Press, 1990.

\bibitem{2014CVX}
M.~Grant and S.~Boyd, {\em CVX: Matlab software for disciplined convex
  programming, version 2.1}.
\newblock Mar. 2014.

\bibitem{huang2018energy}
C.~Huang, G.~C. Alexandropoulos, A.~Zappone, M.~Debbah, and C.~Yuen, ``Energy
  efficient multi-user {MISO} communication using low resolution large
  intelligent surfaces,'' in {\em Proc. IEEE Global Commun. Conf. (Globecom)
  Wkshps}, Waikoloa, Hawaii, USA, Dec. 2018.

\bibitem{pan2020multicell}
C.~Pan, H.~Ren, K.~Wang, W.~Xu, M.~Elkashlan, A.~Nallanathan, and L.~Hanzo,
  ``Multicell {MIMO} communications relying on intelligent reflecting
  surfaces,'' {\em IEEE Trans. Wireless Commun.}, May. 2020.

\bibitem{cover2012elements}
T.~M. Cover and J.~A. Thomas, {\em Elements of Information Theory}.
\newblock New York: Wiley, 1991.

\end{thebibliography}
        
\end{document}